\documentclass{sn-jnl}
\usepackage[table, dvipsnames]{xcolor}
\usepackage{enumerate}

\usepackage{graphicx}%
\usepackage{multirow}%
\usepackage{amsfonts}%
\usepackage{amsthm}%
\usepackage{mathrsfs}%
\usepackage[title]{appendix}%
\usepackage{textcomp}%
\usepackage{manyfoot}%
\usepackage{booktabs}%
\usepackage{algorithmicx}%
\usepackage{listings}%

\usepackage{thmtools,thm-restate}

\usepackage{pgf, tikz}
\usetikzlibrary{arrows, automata}

\usepackage{epigraph}

\usepackage{colonequals}
\usepackage{amssymb}
\usepackage{amsmath}
\usepackage{paralist}
\usepackage{adjustbox}

\usepackage{algorithm}
\usepackage[noend]{algpseudocode}

\usepackage{stackengine}

\usepackage{mathrsfs}
\usepackage{tikz}
\usepackage{tikz-cd}
\usetikzlibrary{shapes}
\usetikzlibrary{trees}
\usepackage{forest}
\usepackage{pstricks}
\usetikzlibrary{positioning,arrows,calc,matrix,shapes.multipart}
\tikzset{
modal/.style={>=stealth’,shorten >=1pt,shorten <=1pt,auto,node distance=1.5cm,
semithick},
world/.style={circle,draw,minimum size=0.5cm,fill=gray!15},
point/.style={circle,draw,inner sep=0.5mm,fill=black},
reflexive above/.style={->,loop,looseness=7,in=120,out=60},
reflexive below/.style={->,loop,looseness=7,in=240,out=300},
reflexive left/.style={->,loop,looseness=7,in=150,out=210},
reflexive right/.style={->,loop,looseness=7,in=30,out=330}
}

\usepackage{hyperref}
\usepackage[capitalise]{cleveref}

\usepackage{tabularx}

\usepackage[shortlabels]{enumitem}

\usepackage{stmaryrd}

\usepackage{footnotehyper}

\newcommand{\ce}{\colonequals}

\newcommand{\pre}{\mathsf{pre}}

\renewcommand{\phi}{\varphi}
\renewcommand{\epsilon}{\varepsilon}
\newcommand{\sink}{{\!}\mathop{-}\!1}
\newcommand{\hB}{\widehat{B}}

\newcommand{\agents}{\mathcal{A}}

\newcommand{\prop}{\mathit{Prop}}
\newcommand{\atoms}{\prop}
\newcommand{\underQ}{\underline{Q}}
\newcommand{\nsr}{\mathrm{nsr}}

\newcommand{\agseq}[0]{\mathit{AgSeq}}

\newcommand{\bisim}{\mathscr{B}}

\newcommand{\rootp}{RootP}
\newcommand{\rootw}{RootW}
\newcommand{\rootwnsr}{RootW_{\nsr}}

\newcommand{\calM}{\mathcal{M}}
\newcommand{\calU}{\mathcal{U}}
\newcommand{\calF}{\mathcal{F}}

\newcommand{\wacc}{\mathit{WAcc}_{\nsr}}

\newcommand{\ta}{\mathsf{ta}}

\usepackage{subcaption}
\usetikzlibrary{arrows,shapes,decorations,automata,backgrounds,petri}
\usetikzlibrary {arrows.meta,bending,positioning}
\tikzstyle{m}=[circle, thin, draw,
minimum size=14mm,inner sep=3pt]
\tikzstyle{n}=[circle, thin, draw,
minimum size=20mm,inner sep=3pt]
\tikzstyle{nod}= [circle, draw,inner sep=0pt, minimum size=0.5cm ]
\tikzset{
reflexive left/.style={->,loop,looseness=10,in=160,out=220},
reflexive right/.style={->,loop,looseness=10,in=20,out=320}
}

\newtheorem{theorem}{Theorem}
\newtheorem{lemma}[theorem]{Lemma}
\newtheorem{corollary}[theorem]{Corollary}%
\theoremstyle{remark}
\newtheorem{example}[theorem]{Example}%
\theoremstyle{definition}
\newtheorem{definition}[theorem]{Definition}%

\begin{document}
\setlength{\parskip}{0pt}
\title{Consistent Update Synthesis via Privatized Beliefs
}

\author{\fnm{Thomas} \sur{Schlögl}}\email{tschloegl@ecs.tuwien.ac.at}
%
\author{\fnm{Roman} \sur{Kuznets}}\email{rkuznets@ecs.tuwien.ac.at}
%
\author{\fnm{Giorgio} \sur{Cignarale}}\email{giorgio.cignarale@tuwien.ac.at}
%
\affil{\orgdiv{Institute of Computer Engineering}, \orgname{TU Wien}, \orgaddress{\city{Vienna}, 
\country{Austria}}}

\abstract
{
    Kripke models are an effective and widely used tool for representing epistemic attitudes of agents in multi-agent systems, including distributed systems.
    Dynamic Epistemic Logic (DEL) adds communication in the form of model transforming updates.
    Private communication is key in distributed systems as processes exchanging (potentially corrupted) information about their private local state should not be detectable by any other processes. This focus on privacy clashes with the standard DEL assumption for which updates are applied to the whole Kripke model, which is usually commonly known by all agents, potentially leading to information leakage.
    In addition, a commonly known model cannot minimize the corruption of agents' local states due to fault information dissemination.
    The contribution of this paper is twofold:
    (I) To represent leak-free agent-to-agent communication, we introduce a way to synthesize an action model which stratifies a pointed Kripke model into private agent-clusters, each representing the local knowledge of the processes:
    Given a goal formula $\varphi$ representing the effect of private communication, we provide a procedure to construct an action model that (a) makes the goal formula true, (b) maintain consistency of agents' beliefs, if possible, without causing ``unrelated'' beliefs (minimal change) thus minimizing the corruption of local states in case of inconsistent information. 
    (II) We introduce a new operation between pointed Kripke models and pointed action models called \emph{pointed updates} which, unlike the product update operation of DEL, maintain only the subset of the world-event pairs that are reachable from the point, without unnecessarily blowing up the model size.\footnote{This research was funded in whole or in part by the Austrian Science Fund (FWF) project ByzDEL [\href{https://doi.org/10.55776/P33600}{10.55776/P33600}]. 
}

}


\maketitle

\section{Introduction}

\emph{Epistemic logic}~(EL)~\cite{Hin62} has been extremely successful in modeling epistemic and doxastic attitudes of agents and groups in multi-agent systems, including distributed systems~\cite{FagHalMosVar95}.
\emph{Dynamic epistemic logic}~(DEL)~\cite{plaza,ditmarsch2007dynamic} upgrades~EL by introducing model transforming modalities called \emph{updates}.
Relational structures such as \emph{action models} in Action Model Logic~(AML) and \emph{arrow update models} in the Generalized Arrow Update Logic~(GAUL)~\cite{GAUL} are used
to represent the evolution of agents' uncertainty under information change in complex communication scenarios.
GAUL and AML have proved to be equally \emph{update expressive}~\cite{AUL-synth}.
Thus, without loss of generality, we use the term ``update~models'' to refer to either action models of AML or arrow update models of GAUL.
In the update \emph{synthesis} task, the aim is to (i) find whether there exists an update model that makes a given goal formula $\varphi$ true and (ii) construct that update model from $\varphi$~\cite{AUL-synth}.

Existing synthesis methods typically work in a language extended with quantifiers over updates, such as the Arbitrary Action Modal Logic~\cite{Hales} and Arbitrary Arrow Update Modal Logic~\cite{AUL-synth} and do not address the issue of minimal change as a result of the update.
While it is possible to construct AML or GAUL update models representing completely private communication~\cite{sep-dynamic-epistemic} (see, e.g., \cref{BOOM_1}), there is no standardized update synthesis procedure for it.
In addition, most existing synthesis methods 
do not address belief consistency preservation~\cite{AUL-synth}.

Our paper is further motivated by two different yet intertwined issues:
\begin{itemize}
\item As argued by Artemov~\cite{Art20LFCS}, in multi-agent settings, common knowledge of the model (CKM) is required by agents in order to compute higher-order beliefs of other agents. While his argument focuses on uncertainties about facts, it can also be extended to agents' uncertainty about attitudes of other agents. The underlying problem is that, in multi-agent Kripke models, there is an implicit ontological distinction between two kinds of possible worlds: (a) worlds that are actually possible (AP) and (b) worlds that are only virtually possible (VP). While the former worlds constitute, for a given agent, the arena in which the actual world might lie, the latter kind of worlds are considered only to the extent of computing other agent's beliefs.
Artemov argues that without the CKM (comprising both APs and VPs), agents would not be able to compute such higher-order beliefs.
Furthermore, avoiding the CKM assumption improves tolerance against local state corruption: upon receiving corrupted information from a faulty agent, a (correct) agent might only deem the local state of the sender as corrupted, instead of being forced to consider a larger part of the accessible Kripke model inconsistent.

\item Update models do not naturally represent private agent-to-agent communication, as the product update operation typical of DEL is applied in principle to the full product of all the world-event pairs whose preconditions are matched~\cite{ditmarsch2007dynamic}. In this sense, these updates are applied globally to the whole model, potentially leading to information leakage. The problem is also identified by Herzig, as ``it is not easy to come up with a meaningful notion of a speaker''(\cite{Herzig}), highlighting the fact that local communication might have (undesired) effects on a global model.
\end{itemize}

Our solution to both issues is to change the structure of the Kripke models under consideration.
While Artemov proposes to avoid the CKM by changing the definition of truth for knowledge~\cite{Art20LFCS}, we instead suggest a stratified Kripke model structure\footnote{More precisely, the kind of structures that will be introduced for privatization are directed acyclic graphs.}, without changing any basic definitions: by stratifying a Kripke model into a pointed \emph{privatized Kripke model} we detach each agent's view of the model from any other agent's view of it. 
In our stratified models, not only we clearly distinguish between APs and VPs for any agent, but we also avoid the CKM assumption by letting each agent access only a limited and private part of it\footnote{A formal definition of privatization is provided in \cref{sec:privatization}.}, without losing the ability to reason about higher-order beliefs of other agents of arbitrary depths.
At the same time, while we do not solve the broader issue of agency identified by Herzig, we propose an alternative update operation called \emph{pointed updates} which does not apply to the full product of the Kripke model and the action model. In particular, we exploit the novel stratification method (and the distinction between APs and VPs) to discern when to form a world-event tuple and when not. In this sense, our updates are not applied to the whole model, as they are progressively applied only to those worlds that match a certain structure (other than a certain precondition).
The proposed stratified structure complies with the crucial notion of \textit{local view} of an agent, typical of distributed systems, which is only a limited portion of the \textit{global view} of the system, usually inaccessible to agents. The stratification not only represents the idea of each agent having a partial view of a Kripke model, but also distinguishes the actual world from any other worlds, by making it inaccessible. 
In other words, we endorse a \textit{fallibilistic} assumption, that is, no agent can be sure that the global state of the system is captured in their limited view of it.
Naturally, agents might have an accurate view of the system during execution, represented in our model by worlds propositionally equivalent to the actual world that are however accessible only in agents' (respective) private clusters.

The aim of this paper is to provide a novel action model synthesis mechanism for AML designed for enforcing private beliefs specified by a quantifier-free goal formula while preserving the consistency of agents' beliefs whenever possible.
Our solution to the leak-free consistent update synthesis task accounts for a limited range of goal formulas, restricted to deterministic belief increase only, i.e. to conjunctions of (positive) modal operators.
This limitation is due to the fact that preserving minimality and consistency becomes highly problematic for unrestricted goal formulas: one the one hand, negations of belief operators might introduce ignorance, potentially contradicting previously held beliefs. On the other hand, disjunctions of modal operators have multiple realizations, making the update non-deterministic.
Addressing the remaining cases (and their interactions) is left for future work.

Finally, the proposed update mechanism is generally more efficient w.r.t.~AML and GAUL updates as it deletes all states not reachable from the actual world, making the growth in size of the model at worst linear in the size of the goal formula after one update and decreasing starting from the second update based on the same modal syntactic tree.
By contrast, it is well known that repeated applications of AML and GAUL updates lead to the exponential growth of the model. It is also known that the model checking problem is PSPACE-complete \cite{DBLP:conf/tark/AucherS13}.

\emph{Related work}. A prominent approach for dealing dynamically with beliefs in Kripke terms is the Dynamic Belief Revision~\cite{DBR,vDRevoc} account, where agents' conditional beliefs are expressed via a plausibility relation ranging over virtual states.
This account, however, cannot represent truly \emph{private} updates, as its models encode both plausibility and epistemic relations that would be revealed by the CKM assumption.
In recent years, different answers to the limitations of the CKM assumption in the standard semantics for EL and DEL have been proposed, for example by changing accessibility relations~\cite{Art20LFCS} or by using belief bases as primitives~\cite{Lorini,Lorini2020-LORREL}.
Our approach, on the other hand, performs private, leakage-free consistent synthesis by structuring standard Kripke models into exclusively accessible parts that are not known to other agents, let alone commonly known. \looseness=-1

\emph{Paper organization}. 
We start by giving a motivating example in \cref{sec:formal_intro}, as well as introducing the basic definitions of DEL and the crucial new definition of pointed update.
In \cref{subsec:Priv} we propose a solution to the leak-free consistent update synthesis for a limited range of goal formulas and we show its fruitfulness by applying it to our motivating example.
In \cref{sec:privatization} we illustrate the properties of the synthesized action models, in particular w.r.t. privatization properties, crucial for avoiding information leakage while preserving consistency.
 Finally, conclusions are provided in \cref{subsec:concl}.

\section{Formal Preliminaries and Motivation} \label{sec:formal_intro}
We use the standard doxastic multimodal language 
\[
\varphi \coloncolonequals p \mid\neg \varphi \mid (\varphi \wedge \varphi) \mid B_i \varphi
\]
where $i \in \agents=\{1,\dots,n\}$ (for $ n>1$) and $p \in \prop$. Formula $B_i \varphi$ means \emph{agent $i$ believes $\varphi$ (to be true)}; note that we generally do not assume the factivity of beliefs.
The other Boolean connectives are defined as usual and $\hB_i \ce \lnot B_i \lnot$ representing the dual modality \emph{considers possible}.

\begin{example}[Balder--Loki--Thor example]
\label{Ex:BLT}
Teenage brothers Balder, Loki, and Thor prepare for an exam that, as is commonly known among them,  contains one question asking to decide whether $p$ or $\neg p$ is the case. Suppose the correct answer is $\neg p$.
The initial model $\calM$ is the left model of~\cref{fig:LT_a}. While the three are studying, Loki, unbeknownst to Thor, tries to trick Balder: Loki lies to Balder that he has overheard Thor boasting to know the correct answer to be~$p$.
Balder seems to dismiss Loki's claim as a ruse, leaving Loki thinking his trick had no effect.
In truth, however, Balder does believe Loki and is now under the impression that Loki agrees with Thor that the answer is~$p$.
Balder himself, on the other hand, is not going to take Thor's word on it, given Thor's propensity to boast and act rashly.
Thus, using $b$, $l$, and $t$ for the brothers, Loki's trick should change Balder's beliefs to achieve the \emph{goal formula} 
\begin{equation}
\label{blt_goal}
\varphi= B_b (B_t p \land B_l B_t p \land B_l p)
\end{equation}
without either affecting Loki's or Thor's beliefs or  making $B_b p$ true.

Since ignorance of the correct answer is common belief in the initial model~$\calM$,  the public announcement of $\varphi$ would cause everybody's beliefs to become inconsistent, whereas a private message $B_t p \land B_l B_t p \land B_l p$ would cause the same inconsistency, but for Balder's beliefs only.
This happens whether one uses the world-removing or arrow-removing updates.

Instead, we propose an update method that stratifies the initial model~$\calM$ into several clusters  representing individual beliefs and updating only some of these clusters, depending on the modal structure of the goal formula, resulting in model~$\calM^U$ in~\cref{fig:LT_a}, where updated clusters are represented by light-gray rectangles.
Cluster~0 contains only the actual world~$v$, which  no agent considers possible.
Cluster~$-1$ (the \emph{sink}) is the copy of the initial model representing (higher-order) beliefs of agents unaffected by the message, including Loki's and Thor's beliefs.
In particular, the ignorance of the correct answer is still common belief between Loki and Thor.
Balder is also ignorant of the correct answer: his beliefs are represented by cluster~$1$, so $\calM^U, v \nvDash B_b p \lor B_b \neg p$. But his ignorance is not anymore common with Loki and Thor. Indeed,
according to cluster $2$, which represents Balder's beliefs about Thor's  beliefs,  $\calM^U, v \vDash B_b B_t p$. 
Cluster~$3$ plays the same role for Balder's beliefs about Loki's beliefs, making $\calM^U, v \vDash B_b B_l p$. Thus, while himself being ignorant of the answer, Balder thinks that Loki and Thor have chosen for themselves.
Finally, cluster~$4$ is Balder's beliefs about Loki's beliefs about Thor's beliefs so that $\calM^U, v \vDash B_bB_lB_t p$.
Overall, $\calM^U, v \vDash \varphi$.
The $b$-labeled double arrow from cluster $3$ to the sink~$-1$ represents two $b$-labeled arrows from the only world of cluster~$3$ to each of the two worlds of the sink and ensures that Balder believes that Loki believes that Balder has not changed his beliefs,  $\calM^U, v \vDash B_b B_l (\neg B_b p \land \neg B_b \neg p)$. Similarly, the $b,l$-labeled double arrows from clusters~$2$~and~$4$ to the sink signify that in no scenario where Thor has chosen an answer for himself, does he expect others to be follow his example (or even be aware of his choice).
It is easy to see that no agent has inconsistent beliefs, $\calM^U \vDash \neg B_b \bot \land \neg B_l \bot \land \neg B_t \bot$, and that Loki's and Thor's beliefs are the same as in the initial model, $M, v \vDash B_a \psi$ if{f} $\calM^U, v \vDash B_a \psi$ for $a \in \{l,t\}$.

\end{example}
\begin{figure}[t]
    \centering
    \resizebox{13cm}{!}{%
    
\begin{subfigure}{0.8\textwidth}
\begin{tikzpicture}
\begin{scope}[xshift=1cm,>=stealth]
    \node[m] (1) {$p$};
    \node[m,fill=gray!40,label=45:{$v$}] (2) [right = 2cm  of 1]  {$\neg p$};
    \node (g) [right = 1cm  of 1] {};
    \node [above=1cm of g] {\textbf{\huge{$\calM$}}};
    \node (n) [below=5cm of g] {};
    
    \path[]
        (1) edge [<->] node[above] {$b$,$l$,$t$} (2)
            edge [loop left, ->] node {$b$,$l$,$t$} (1)
        (2) edge [loop right, ->] node {$b$,$l$,$t$} (2)
            ;
    \end{scope}
\end{tikzpicture}
\label{4}
\end{subfigure}

\hfill

\begin{subfigure}{0.9\textwidth}
\begin{tikzpicture}
    \begin{scope}[xshift=-6cm,>=stealth, yshift=3cm]
    \node[m,fill=gray!50,label=135:{0},label=45:{$v$}] (0) {$\neg p$};
    \node (g) [below = 1cm  of 0] {};
    \node [left=1.5cm of 0] {\textbf{\huge{$\calM^U$}}};
    \node (n) [below = 1cm  of 0] {};
    \node[m,label=135:{4}] (2a) [right = 1.5cm  of n] {$p$};
    \node[m] (2b) [right = .5cm  of 2a] {$\neg p$};
    \node[m,,label=135:{3}] (3a) [below = 1.5cm  of 2a] {$p$};
    \node[m,label=135:{2}] (3b) [right = .7cm  of 3a] {$p$};
    \node[m,label=135:{1}] (4a) [below = 1.5cm  of 3b] {$p$};
    \node[m,label=135:{$-1$}] (6a) [left = 3.5cm  of 4a]  {$p$};
    \node[m] (6b) [right = .5cm  of 6a]  {$\neg p$};
    \node (g1) [right = 0.01cm  of 0] {};
    \node (g2) [above = 0.01cm  of 2a] {};
    \node (g3) [below = 0.01cm  of 2a] {};
    \node (g4) [above = 0.01cm  of 3a] {};
    \node (g5) [below = 0.01cm  of 2b] {};
    \node (g6) [above = 0.01cm  of 3b] {};
    \node (g7) [below = 0.01cm  of 3b] {};
    \node (g8) [above = 0.01cm  of 4a] {};
    \node (g9a) [below = 0.01cm  of 3a] {}; 
    \node (g9b) [left = 0.5cm  of g9a] {};
    \node (g10a) [above = 0.01cm  of 6a] {}; 
    \node (g10b) [right = .7cm  of g10a] {};
    \node (g11a) [below = 0.01cm  of 3b] {}; 
    \node (g11b) [left = 0.5cm  of g11a] {};
    \node (g12) [right = 2cm  of g10a] {};
    \node (g13) [left = 0.01cm  of 4a] {}; 
    \node (g14a) [right = 0.3cm  of g12] {}; 
    \node (g14b) [below = 0.6cm  of g14a] {};
    \node (g15a) [below = 0.01cm  of 0] {};
    \node (g15b) [left = 0.3cm  of g15a] {};
    \node (g16a) [right = 1cm  of g10a] {};
    \node (f1) [right = 0.01cm  of 2b] {};
    \node (f2) [right = 0.01cm  of 3b] {};
    \node (f3) [right = 0.01cm  of 4a] {};
    \node (f4) [left = 0.01cm  of 3a] {};
    \node (f5) [left = 0.01cm  of 6a] {};
    
    \draw[-{Implies},double,line width=2pt] (g1) -- (g2) node[midway,above] {$b$};
    \draw[-{Implies},double,line width=2pt] (g3) -- (g4) node[midway,left] {$t$};
    \draw[-{Implies},double,line width=2pt] (g5) -- (g6) node[midway,right] {$l$};
    \draw[-{Implies},double,line width=2pt] (g7) -- (g8) node[midway,right] {$t$};
    \draw[-{Implies},double,line width=2pt] (g9b) -- (g10b) node[midway,right] {$b$,$l$};
    \draw[-{Implies},double,line width=2pt] (g11b) -- (g12) node[midway,above] {$b$};
    \draw[-{Implies},double,line width=2pt] (g13) -- (g14b) node[midway,above] {$b$,$l$};
   \draw[-{Implies},double,line width=2pt] (g15b) -- (g10a) node[midway,right] {$l$, $t$};
        
    \path[->] (f1) edge[-{Implies},double,line width=2pt,reflexive right] node[right] {$b$} (f1);
    \path[->] (f2) edge[-{Implies},double,line width=2pt,reflexive right] node[right] {$l$} (f2);
    \path[->] (f3) edge[-{Implies},double,line width=2pt,reflexive right] node[right] {$t$} (f3);
    \path[->] (f4) edge[-{Implies},double,line width=2pt,reflexive left] node[left] {$t$} (f4);
    \path[->] (f5) edge[-{Implies},double,line width=2pt,reflexive left] node[left] {$b$, $l$, $t$} (f5);

    \end{scope}
\begin{pgfonlayer}{background}
\filldraw [line width=4mm,line join=round,black!10]
      (0.north  -| 0.east)  rectangle (0.south  -| 0.west);
 \filldraw [line width=4mm,line join=round,black!10]
      (2a.north  -| 2b.east)  rectangle (2a.south  -| 2a.west);
    \filldraw [line width=4mm,line join=round,black!10]
      (3a.north  -| 3a.east)  rectangle (3a.south  -| 3a.west);
       \filldraw [line width=4mm,line join=round,black!10]
      (3b.north  -| 3b.east)  rectangle (3b.south  -| 3b.west);
       \filldraw [line width=4mm,line join=round,black!10]
      (4a.north  -| 4a.east)  rectangle (4a.south  -| 4a.west);
      \filldraw [line width=4mm,line join=round,black!10]
      (6a.north  -| 6b.east)  rectangle (6a.south  -| 6a.west);
  \end{pgfonlayer}
\end{tikzpicture}
\label{5}
\end{subfigure}
    }
    \caption{Left: Initial pointed Kripke model $\calM$. Right: Updated model $\calM^U$ (the actual world~$v$ is dark gray). A formula $\phi$ within a circle representing a world means that $\phi$~is true in this world. A light-gray rectangle with $n$ in the top left corner is the cluster that is numbered~$n$.
    A double arrow from such cluster $n$ to cluster $m$  represents a set of accessibility arrows with the same agent's label from every world of cluster~$n$ to every world of cluster~$m$.\looseness=-1}
    \label{fig:LT_a}
\end{figure}
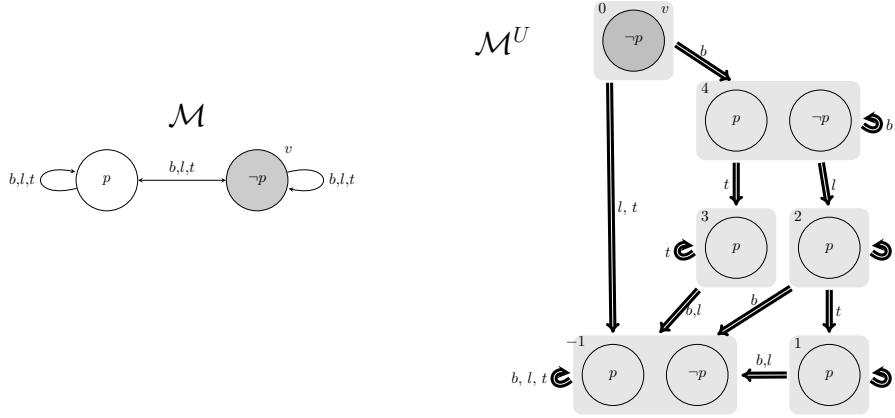

\begin{definition}
\label{def:target}
The set of \emph{target agents} of a modal formula $\phi$ is defined  as follows: $\ta(p) \ce \varnothing$; $\ta(\lnot \phi) \ce \ta(\phi)$; $\ta (\phi \land \psi) \ce \ta(\phi) \cup \ta(\psi)$; finally $\ta (B_i \phi) \ce \{i\}$. 
 \end{definition}
 E.g., $\ta\bigl( B_b (B_t p \land B_l B_t p \land B_l p)\bigr) = \{b\}$ even though $\ta(B_t p \land B_l B_t p \land B_l p) = \{t,l\}$.

We introduce the standard definitions of DEL \cite{ditmarsch2007dynamic}:

\begin{definition}[Kripke frame]
A Kripke frame for the set of agents $\agents$ is a pair $\langle W, R\rangle$ where $W \ne \varnothing$ is a non-empty set, called \emph{domain} and $R = (R_1, \ldots,  R_n)$ consists of  binary \emph{accessibility relations} $R_i \subseteq W \times W$ on~$W$.
\end{definition}

\begin{definition}[Pointed Kripke model] 
\label{def:kripke}
	For a set of agents $\agents$, a \emph{Kripke model} is a triple $\calM = \langle S, R, V\rangle$ where $\langle S, R\rangle$ is a Kripke frame with domain~$S$ consisting of  
 \emph{possible worlds} or \emph{states},
and the \emph{valuation function} $V \colon \atoms \to 2^S$ determines the set $V(p)\subseteq S$ of possible worlds where an atom~$p \in \prop$ is true. A \emph{pointed Kripke model} is a pair $(\calM,w)$ where 
 $w \in S$ represents the \emph{actual world/state}.

\emph{Truth at world $w$ of model $\calM$} is determined by $\calM, w \vDash p$ if{f} $w \in V(p)$. Boolean connectives are defined as usual. $\calM, w \vDash B_i \varphi$  if{f}  $\calM, w' \vDash \varphi$ for all  $w'\in S$ such that $w R_i w'$.
A formula $\varphi$ is \emph{false at world $w$}, $\calM,w \nvDash \varphi$, if{f} it is not true at $w$.\looseness=-1
\end{definition}

\begin{definition}[Pointed action model]
\label{def:AM}
    For a set of agents $\agents$, an \emph{action model} is a triple $\calU = \langle E, Q, \pre \rangle$ where $\langle E, Q\rangle$ is a Kripke frame with domain $E$ consisting of \emph{action points} or \emph{events}, and    the \emph{precondition function} $\pre: E \rightarrow \mathcal{L}$ assigns the precondition  $\pre(\beta)\in \mathcal{L}$ that is necessary for an event $\beta \in E$ to happen. A \emph{pointed action model} is a pair $(\calU,\alpha)$ where $\alpha \in E$ represents the \emph{actual event/action point}.
\end{definition}

\begin{definition}[Product update]
\label{def:prod_up}
    The \emph{(restricted modal) product update} of  a Kripke model $\calM = \langle S,R,V \rangle$ with an action model $\calU = \langle E, Q, \pre \rangle$ is   a Kripke model \mbox{$\calM \otimes \calU \ce \langle S',R',V' \rangle$} where
\begin{itemize}
    \item $S' \ce \{(v, \beta) \in S \times E \mid \calM, v \vDash \pre(\beta) \}$,
    \item 
    $R'_i \ce \left\{\bigl((v, \beta),(u, \gamma)\bigr)\in S'\times S' \mid (v, u) \in R_i \text{ and } (\beta, \gamma) \in Q_i\right\}$,
    \item 
    $V'(p) \ce \{(v, \beta) \in S' \mid v \in V(p)\}$.
\end{itemize}
If $S' = \varnothing$, the product update is undefined.
\end{definition}

A product update provides the semantics for a communication scenario represented by a pointed action model $(\calU,\alpha)$: formula $\phi$ is true after the communication $\calU$, i.e., $\calM, w \vDash [\calU,\alpha]\phi$, if{f} $\calM \otimes \calU, (w,\alpha) \vDash \phi$ or $(w,\alpha) \notin S'$.

It is well known that repeatedly applying product updates for multi-round communication can cause exponential blow up of the domain of the model, even for simple models and action models, see  Figs.~\ref{BOOM_1}--\ref{BOOM_3}. Indeed, every updated model has exactly one world where $p$ is false. Hence, a further update of a model of $N$~worlds with~$\calU$ creates a model with $2N-1$~worlds.

The standard (often implicit) solution is to reduce the size of the model by deleting all worlds that are unreachable from the resulting actual world.
This operation is technically sound, as deleting all these unreachable worlds yields a model that is  bisimilar  and, hence, modally equivalent~\cite{blackburn_rijke_venema_2001}.

We introduce a new definition of \textit{pointed update} that incorporates this world-removing operation formally and explicitly:

\begin{definition}[Pointed update]
\label{def:pointed_update}
Let $(\calM,w)=(\langle S, R, V\rangle, w)$ be a pointed Kripke model and $(\calU,\alpha) = (\langle E, Q, \pre\rangle, \alpha)$ be a pointed action model. If \mbox{$\calM, w \vDash \pre(\alpha)$}, then we define the updated pointed Kripke model $\bigl(\calM\odot\calU, (w,\alpha)\bigr)$ where \mbox{$\calM\odot\calU \ce \langle S^\calU, R^\calU, V^\calU\rangle$} as follows.
Let 
$
T^\calU \colonequals \bigl\{(v,\beta) \in S \times E \mid  \calM, v \vDash \pre(\beta)\bigr\}$.
Note that $(w,\alpha) \in T^\calU$.
\begin{itemize}
\item The domain $S^\calU$ is the smallest subset of $T^\calU$ containing 
$(w,\alpha)$ such that:\\
 if $(v,\beta) \in S^\calU$, $(u,\gamma) \in T^\calU$, and both $v R_i u$ and $\beta Q_i \gamma$ for some agent $i$, then $(u,\gamma) \in S^\calU$;
\item $R^\calU_i \ce \left\{\bigl((v, \beta),(u, \gamma)\bigr)\in S^\calU\times S^\calU \mid (v, u) \in R_i \text{ and } (\beta, \gamma) \in Q_i\right\}$;
\item $V^\calU(p) \ce \left\{(v, \beta)\in S^\calU \mid v \in V(p)\right\}$.
\end{itemize}
\end{definition}

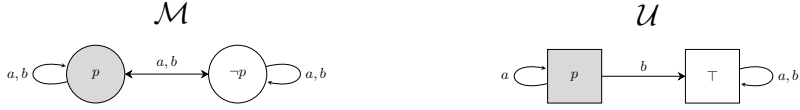
\begin{figure}[t]
    \centering
    \resizebox{13cm}{!}{%
    
\begin{subfigure}{0.9\textwidth}
\begin{tikzpicture}
\begin{scope}[xshift=1cm,>=stealth]
    \node[m, fill=gray!30] (1) {$p$};
    \node[m] (2) [right = 2cm  of 1]  {$\neg p$};
    \node (g) [right = 1cm  of 1] {};
    \node [above=1cm of g] {\textbf{\huge{$\calM$}}};
    
    \path[]
        (1) edge [<->,{Stealth[scale width=1.5]}-{Stealth[scale width=1.5]}] node[above] {$a,b$} (2);
    \path[->, -{Stealth[scale width=1.5]}]    
        (1) edge [loop left] node {$a,b$} (1)
        (2) edge [loop right] node {$a,b$} (2)
            ;
    \end{scope}
\end{tikzpicture}
\label{4}
\end{subfigure}

\hfill

\begin{subfigure}{0.9\textwidth}
\begin{tikzpicture}
    \begin{scope}[xshift=-6cm,>=stealth, yshift=3cm]
    \node[fill=gray!30,minimum size=1.3cm,draw] (1) {$p$} ;
    \node[minimum size=1.3cm,draw] (2) [right = 2cm  of 1]  {$\top$} ;
    \node (g) [right = 1cm  of 1] {};
    \node [above=1cm of g] {\textbf{\huge{$\calU$}}};
 \path[]
        (1) edge [->, -{Stealth[scale width=1.5]}] node[above] {$b$} (2)
            edge [loop left, ->] node {$a$} (1)
        (2) edge [loop right, ->] node {$a,b$} (2)
            ;
    \end{scope}

\end{tikzpicture}
\label{5}
\end{subfigure}
    }
    \caption{Left: Initial (pointed) Kripke model $\calM$ where agent $a$ and $b$ are uncertain about whether $p$~or~$\lnot p$ is true. Right: (Pointed) action model $\calU$ for a private message~$p$ to $a$: agent~$a$  learns $p$, while $b$ believes that nothing happened. A formula $\phi$ within a square representing an event $\beta$ is the precondition for this event, i.e., $\pre(\beta) = \phi$.\looseness=-1}
    \label{BOOM_1}
\end{figure}

\begin{figure}[t]
    \centering
    \resizebox{13cm}{!}{%
    
\begin{subfigure}{0.9\textwidth}
\begin{tikzpicture}
\begin{scope}[xshift=1cm,>=stealth]
    \node[m] (1) {$p,\top$};
    \node[m] (2) [right = 2cm  of 1]  {$\neg p,\top$};
    \node (g) [right = 1cm  of 1] {};
    \node [above=1cm of g] {\textbf{\huge{$\calM \otimes \calU$}}};
    \node[m,fill=gray!30] (3) [below=1.5cm of g] {$p,p$};
    
    \path[]
        (1) edge [<->, {Stealth[scale width=1.5]}-{Stealth[scale width=1.5]}] node[above] {$a,b$} (2)
            edge [loop left, ->] node {$a,b$} (1)
        (2) edge [loop right, ->] node {$a,b$} (2)
        (3) edge [->, -{Stealth[scale width=1.5]}] node[above] {$b$} (1)
            edge [->, -{Stealth[scale width=1.5]}] node[above] {$b$} (2)
            edge [loop right, ->] node {$a$} (3)
            ;
    \end{scope}
\end{tikzpicture}
\label{4}
\end{subfigure}

\hfill

\begin{subfigure}{0.9\textwidth}
\begin{tikzpicture}
    \begin{scope}[xshift=-6cm,>=stealth, yshift=3cm]
    \node[m] (1) {$p,\top, \top$};
    \node[m] (2) [right = 2cm  of 1]  {$\neg p,\top, \top$};
    \node (g) [right = 1cm  of 1] {};
    \node [above=1cm of g] {\textbf{\huge{$(\calM \otimes \calU)\otimes \calU$}}};
    \node[m,fill=gray!30] (3) [below=1.5cm of g] {$p,p,p$};
    \node[m] (4) [right = 1.5cm  of 3]  {$p,p,\top$};
    \node[m] (5) [left = 1.5cm  of 3]  {$p,\top,p$};
    
    \path[]
        (1) edge [<->, {Stealth[scale width=1.5]}-{Stealth[scale width=1.5]}] node[above] {$a,b$} (2)
        (1) edge [loop left, ->] node {$a,b$} (1)
        (2) edge [loop right, ->] node {$a,b$} (2)
        (3) edge [->, -{Stealth[scale width=1.5]}] node[above] {$b$} (1) edge [->, -{Stealth[scale width=1.5]}] node[above] {$b$} (2)
            edge [loop right, ->] node {$a$} (3)
        (4) edge [loop right, ->] node {$a$} (4)
        (4)    edge [->, -{Stealth[scale width=1.5]}] node[above] {$b$} (1)
        (4)    edge [->, -{Stealth[scale width=1.5]}] node[above] {$b$} (2)
        (5) edge [loop left, ->] node {$a$} (5)
        (5)    edge [->, -{Stealth[scale width=1.5]}] node[above] {$b$} (1)
        (5)    edge [->, -{Stealth[scale width=1.5]}] node[above] {$b$} (2)
            ;
    \end{scope}

\end{tikzpicture}
\label{5}
\end{subfigure}}
    \caption{Left: First product update $\calM \otimes \calU$ of Kripke model $\calM$ with action model $\calU$. Right: Second product update $(\calM \otimes \calU) \otimes \calU$ with the same action model~$\calU$. $p$ is true in all worlds that start from~$p$ and false in all worlds that start from $\lnot p$.\looseness=-1}
    \label{BOOM_2}
\end{figure}
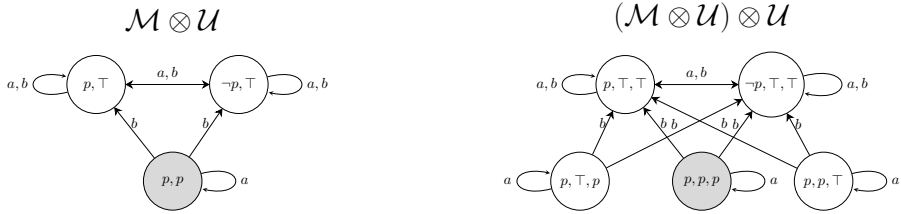

\begin{figure}
    \centering
    \resizebox{11cm}{!}{%
    \begin{tikzpicture}    
    \node[n] (1) {$\small{p,\top, \top, \top}$};
    \node[n] (2) [right = 2cm  of 1]  {$\tiny{\neg p,\top, \top, \top}$};
    \node (g) [right = 1cm  of 1] {};
    \node [above=1cm of g] {\textbf{\Large{$\bigl((\calM \otimes \calU)\otimes \calU\bigr)\otimes \calU$}}};
    \node[n,fill=gray!30] (3) [below=3cm of g] {$p,p,p,p$};
    \node[n] (4) [right = 1cm  of 3]  {$p,\top,p,\top$};
    \node[n] (5) [left = 1cm  of 3]  {$p,\top,p,p$};
    \node[n] (6) [right = 1cm  of 4]  {$p,p,p, \top$};
    \node[n] (7) [left = 1cm  of 5]  {$p,\top,\top,p$};
    \node[n] (8) [left = 1cm  of 7]  {$p,p,\top,p$};
    \node[n] (9) [right = 1cm  of 6]  {$p, p,\top,\top$};
    
    \path[<->, {Stealth[scale width=1.5]}-{Stealth[scale width=1.5]}]
        (1) edge [] node[above] {$a,b$} (2);
    \path[]    
        (1) edge [loop left, ->] node {$a,b$} (1)
        (2) edge [loop right, ->] node {$a,b$} (2)
        (3) edge [->] node[above] {$b$} (1)
            edge [->] node[above] {$b$} (2)
            edge [loop below, ->] node {$a$} (3)
        (4) edge [loop below, ->] node {$a$} (4)
            edge [->] node[above] {$b$} (1)
            edge [->] node[above] {$b$} (2)
        (5) edge [loop below, ->] node {$a$} (5)
            edge [->] node[above] {$b$} (1)
            edge [->] node[above] {$b$} (2)
        (6) edge [loop below, ->] node {$a$} (6)
            edge [->] node[above] {$b$} (1)
            edge [->] node[above] {$b$} (2)
        (7) edge [loop below, ->] node {$a$} (7)
            edge [->] node[above] {$b$} (1)
            edge [->] node[above] {$b$} (2)
        (8) edge [loop below, ->] node {$a$} (8)
            edge [->] node[above] {$b$} (1)
            edge [->] node[above] {$b$} (2)
        (9) edge [loop below, ->] node {$a$} (9)
            edge [->] node[above] {$b$} (1)
            edge [->] node[above] {$b$} (2)
            ;

\end{tikzpicture}
}
    \caption{Third product update $\bigl((\calM \otimes \calU)\otimes \calU\bigr)\otimes \calU$ with the same action model $\calU$.}
    \label{BOOM_3}
\end{figure}
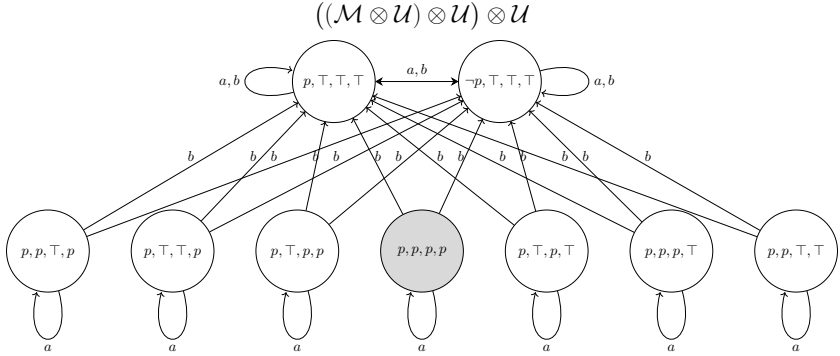

To show that pointed updates produce  the same result as product updates but with smaller models, we use the notion of bisimulation.

\begin{definition}[Bisimulation]
A \emph{bisimulation between Kripke frames} $\langle W, R\rangle$ and $\langle W', R'\rangle$ is a nonempty binary relation $\bisim \subseteq W \times W'$, such that for every $v\bisim v'$ and  $i \in \agents$:
\begin{compactitem}
    \item \textbf{Forth}: if $v R_i u$, then there is $u'\in S'$ such that $v' R_i'u'$ and  $u\bisim u'$;
    \item \textbf{Back}: if $v' R_i' u'$, then there is $u \in S$ such that  $v R_i u$ and $u\bisim u'$.
\end{compactitem}

A \emph{bisimulation between Kripke models} $\calM=\langle S, R,V\rangle$ and $\calM'=\langle S', R',V'\rangle$ is a bisimulation $\bisim\subseteq S \times S'$ between their frames $\langle S, R\rangle$ and $\langle S', R'\rangle$ that additionally satisfies for every $v\bisim v'$ and  $p \in \prop$:
\begin{compactitem}
    \item \textbf{Atoms:} $v \in V(p)$ iff $v' \in V'(p)$.
\end{compactitem}

A \emph{bisimulation between action models} $\calU=\langle E, Q, \pre\rangle$ and $\calU'=\langle E', Q', \pre'\rangle$ is a bisimulation $\bisim\subseteq E \times E'$ between their frames $ \langle E, Q\rangle$ and $\langle E',Q'\rangle$ that additionally satisfies for every $\alpha\bisim \alpha'$:
\begin{compactitem}
    \item \textbf{Pre:} $\pre(\alpha) = \pre'(\alpha')$.
\end{compactitem}

Pointed Kripke models $(\calM, v)$ and $(\calM', v')$ are \emph{bisimilar}, notation $\calM, v \leftrightarroweq \calM', v'$, if{f} there is a bisimulation $\bisim$ between $\calM$ and $\calM'$ such that $v \bisim v'$. 

Pointed Kripke models $(\calM, v)$ and $(\calM', v')$ are \emph{$G$-bisimilar} for a group $G \subseteq \agents$ of agents, notation $\calM, v \leftrightarroweq_G \calM', v'$, if{f}  for any $a \in G$
\begin{compactitem}
    \item \textbf{$G$-Forth}: if $v R_a u$, then there is $u'\in S'$ such that $v' R_a'u'$ and  $\calM,u\leftrightarroweq \calM,u'$;
    \item \textbf{$G$-Back}: if $v' R_a' u'$, then there is $u \in S$ such that  $v R_a u$ and $\calM,u\leftrightarroweq \calM,u'$.
\end{compactitem}
The definitions of bisimilarity and $G$-bisimilarity for pointed action models are analogous. 
\end{definition}

\begin{definition}[Equivalences]
Pointed Kripke models $(\calM, v)$ and $(\calM', v')$  are called \emph{modally equivalent}, notation  $\calM, v \equiv \calM', v'$, whenever $\calM, v \vDash \phi$ if{f} $\calM', v' \vDash \phi$ for all formulas~$\phi$. They are called \emph{$G$-indistinguishable for group $G \subseteq \agents$}, notation  \mbox{$\calM, v \equiv_G \calM', v'$}, whenever $\calM, v \vDash B_a\phi$ if{f} $\calM', v' \vDash B_a\phi$ for all formulas~$\phi$ and agents $a\in G$.
\end{definition}

\begin{theorem}[Bisimilarity implies modal equivalence]
\label{bisme}
\leavevmode
\begin{compactenum}
\item\label{bismeall}
If $\calM, v\leftrightarroweq\calM', v'$, then $\calM, v\equiv\calM', v'$.
\item\label{bismeone}
If $\calM, v\leftrightarroweq_G\calM', v'$, then $\calM, v\equiv_G\calM', v'$.
\end{compactenum}
\end{theorem}
\begin{proof}
The first statement is standard (see~\cite{blackburn_rijke_venema_2001}). The second statement easily follows from the first.
\end{proof}

\begin{theorem}[Product and pointed updates are bisimilar]
\label{th:sameupdates}
    Given a pointed Kripke model $(\calM,w)$ and pointed action model $(\calU,\alpha)$ such that $\calM,w \vDash \pre(\alpha)$, we have $\calM\odot\calU, (w,\alpha) \leftrightarroweq \calM\otimes\calU, (w,\alpha)$.
\end{theorem}

\begin{proof} 
Let $\calM= \langle S, R, V\rangle$, $\calU = \langle E, Q, \pre\rangle$, $\calM\odot\calU = \langle S^\calU, R^\calU, V^\calU\rangle$ and \mbox{$\calM \otimes \calU = \langle S',R',V' \rangle$}. By construction, $(w,\alpha) \in S^\calU \subseteq S'$, and $R^\calU =  R'_{\upharpoonright (S^\calU \times S^\calU)}$, and $V^\calU = V'_{\upharpoonright S^\calU}$.
Hence, it is easy to see that $\bisim \ce \left\{\bigl((v,\beta),(v,\beta)\bigr) \mid (v,\beta) \in S^\calU\right\}$ is the requisite bisimulation.
\end{proof}

\begin{theorem}[Updates preserve bisimilarity]
\label{th:updatepreserve}
    Given bisimilar pointed Kripke model $\calM,w\leftrightarroweq \calM',w'$ and bisimilar pointed action models $\calU,\alpha\leftrightarroweq\calU',\alpha'$ such that \mbox{$\calM,w \vDash \pre(\alpha)$}, we have $\calM\otimes\calU, (w,\alpha) \leftrightarroweq \calM'\otimes\calU', (w',\alpha')$ and \mbox{$\calM\odot\calU, (w,\alpha) \leftrightarroweq \calM'\odot\calU', (w',\alpha')$}. The same holds for $G$-bisimilarity.
\end{theorem}
\begin{proof}
For product updates, the proof can be found in \cite{ditmarsch2007dynamic}. For pointed updates, it then follows from Theorem~\ref{th:sameupdates}. The statements for $a$-bisimilarity easily follow.
\end{proof}

The advantage of pointed updates can be illustrated by the fact that, while repeated product updates of $\calM$ from Fig.~\ref{fig:LT_a} with~$\calU$ from Fig.~\ref{BOOM_1} leads to the exponential growth of the domain, pointed updates yield the three-world pointed Kripke model depicted left in Fig.~\ref{BOOM_2}, no matter how many times the pointed update is performed:
\[
\bigl((\calM \odot \calU) \odot \calU\bigr) \odot \dots \odot \calU = \calM \otimes \calU.
\]

In the following section, we use this mechanism to solve the consistent update synthesis problem.

\section{Consistent Update Synthesis:  Deterministic Belief Increase}
\label{subsec:Priv}

In this section, we propose a solution to the consistent update synthesis task for \emph{goal formulas}~$\varphi$ representing \emph{deterministic belief increase}~(DBI) by  generating a pointed action model $\calU_\phi$ such that the result of a pointed update of any given pointed Kripke model with $\calU_\phi$ satisfies $\phi$ and inconsistent beliefs (including higher-order beliefs) are not introduced  whenever they are possible to avoid.
By deterministic belief increase we mean situations when several agents are prescribed additional beliefs (including higher-order beliefs) without creating alternative ways of fulfilling the prescription. Before giving the formal definition, it should be mentioned that there is always a trivial way of creating new beliefs by making the agents' beliefs inconsistent, i.e., making them believe all statements including the desired ones. However, making agents' reasoning inconsistent does not comport with the ideology of minimal change, nor is productive in terms of correcting agents' incidental false beliefs.
The main novelty in our method of update synthesis is the aim to preserve consistency whenever possible.
Now we give a formal definition of goal formulas representing deterministic belief increase:

\begin{definition}[DBI goal formulas, DBI normal form]%
\label{def:goal_formula}
\emph{DBI goal formulas}~$\varphi$, or \emph{DBI~formulas} for short, are defined by the following BNF:
	\begin{equation} \label{eq:goal_formula}
			\varphi  \coloncolonequals  B_i \xi \mid B_i(\xi \wedge \varphi) \mid (\varphi \wedge \varphi) \mid B_i \varphi
	\end{equation}
where $\xi$ is any purely propositional formula and $i \in \agents$.
Thus, each DBI goal formula is a non-empty conjunction of belief operators. 
\emph{Formulas in DBI normal form} are DBI formulas $\varphi$ obtained by restricting the construction as follows: $B_i \xi$~is always DBI~normal; $B_i\varphi$ and $B_i(\xi \wedge \varphi)$ are DBI normal if{f} $\varphi$ is DBI normal and $i \notin\ta(\varphi)$; $\varphi \wedge \psi$ is DBI normal if{f} $\varphi$ and $\psi$ are DBI normal and $\ta(\varphi) \cap \ta(\psi) = \varnothing$.
\end{definition}

\begin{lemma}
For any DBI goal formula $\phi$, there exists a DBI normal formula~$\phi'$ such that $\mathcal{K}45 \vDash \phi \equiv \phi'$, where $\mathcal{K}45$ is the class of all transitive\footnote{If $vR_i u$ and $uR_i w$ then $vR_i w$.} and euclidean\footnote{If $vR_i u$ and $vR_i w$ then $uR_i w$.} frames.\looseness=-1
\end{lemma}
\begin{proof}
It is sufficient to use the following $ \mathcal{K}45$ logical equivalences: $B_iB_i \theta \equiv B_i \theta$, $B_i(\xi \land B_i \theta) \equiv B_i(\xi \land \theta)$, and $B_i\theta \land B_i \eta \equiv B_i(\theta \land \eta)$.
\end{proof}

\begin{lemma}[$G$-bisimilarity and goal formula preservation]
\label{lem:Gbisimgoal}
If $\calM, v\leftrightarroweq_G\calM', v'$ and $\phi$~is a DBI formula with $\ta(\phi) \subseteq G $, then $\calM, v \vDash \phi$ if{f}  $\calM', v' \vDash \phi$.
\end{lemma}
\begin{proof}
It follows from Theorem~\ref{bisme} $(\textit{2})$
and the fact that $\phi$ is a conjunction of $B_i \psi_i$ for $i \in \ta(\phi) \subseteq G$.
\end{proof}

Without loss of generality, from now on, we only consider formulas in DBI normal form. For instance, formula~\eqref{blt_goal} is a DBI formula but not DBI normal. Its normal form would be $B_b\bigl(B_tp \land B_l(p \land B_t p)\bigr)$.
Given such a DBI normal formula $\varphi$, we now construct a pointed action model~$(\calU_\phi,0)$ such that for any  pointed Kripke model $(\calM, w)$, the pointed update $\calM\odot\calU_\phi$ is defined,  $\calM\odot\calU_\phi, (w,0) \vDash \phi$, while other epistemic differences between $(\calM, w)$ and $\left(\calM\odot\calU_\phi, (w,0)\right)$ are minimized.
Informally, this means that an update should not influence any beliefs except for those explicitly stated in $\phi$ and logically following from $\phi$ based on  the agents' pre-update beliefs \footnote{We adopt the minimality principle in the same spirit of the standard AGM approach \cite{Belief_rev_handbook} for which an update should lead to the loss of as few previous beliefs as possible~\cite{sep-logic-belief-revision}.}.
In particular, if $j \notin \ta(\phi)$, then $j$'s beliefs should not be affected even if $B_j$ occurs in $\phi$: updating somebody else's beliefs about $j$'s beliefs should not affect actual $j$'s beliefs.

\begin{definition}[Consistent DBI update synthesis] \label{def:synth_am1}
Given a formula~$\varphi$ in DBI normal form the pointed action model $(\calU_{\varphi},0) = (\langle E^\phi, Q^\phi, \pre^\phi \rangle, 0)$ is constructed  recursively based on the structure of $\phi$. In all cases, $E^\phi= \{0,-1\}\sqcup D^\phi$ with $0 \ne -1$ for some $\varnothing\ne D^\phi \subseteq\mathbb{N}$, event $0$ representing the actual event, unknowable for all, has no incoming arrows, i.e., not $\alpha Q^\phi_i 0$, and event $-1$ representing status quo/no change  has no outgoing arrows other than reflexive loops, i.e., $-1Q^\phi_i\alpha$ if{f} $\alpha = -1$. \mbox{$\pre^\phi(0) = \pre^\phi(-1) \ce \top$}. Let $\underQ_j^\phi \ce Q^\phi_j \cap \bigl((E^\phi \setminus\{0\}) \times (E^\phi \setminus\{0\})\bigr)$ be the restriction of $Q_j^\phi$ onto $E^\phi \setminus \{0\}$ for any agent $j$. 
\begin{description}
\item[1] If $\phi = B_i \xi$ for a propositional formula $\xi$, then $D^{B_i \xi}=\{m\}$ for a fresh $m\geq1$, $\pre^{B_i \xi}(m) \ce \xi$, $Q^{B_i \xi}_j \ce \{(0,-1), (m,-1), (-1,-1)\}$ for any $j\ne i$,  while at the same time $Q^{B_i \xi}_i \ce \{(0,m), (m,m), (-1,-1)\}$.
\item[2] If $\phi = B_i \psi$ for a DBI normal formula $\psi$, then $\calU_\psi =\langle \{0,-1\}\sqcup D^\psi, Q^\psi, \pre^\psi\rangle$ has already been constructed. Let $1 \leq m\notin D^\psi$ be fresh. Then $D^\phi \ce D^\psi \sqcup \{m\}$, $\pre^\phi\ce \pre^\psi \sqcup \{(m, \top)\}$, i.e., $\pre^\phi$ agrees with $\pre^\psi$ and extends it to $\pre^\phi(m) \ce \top$, and the accessibility relations are defined as follows:  
\begin{align*}
Q^\phi_j &\ce \underQ_j^\psi \sqcup \{(0,-1)\}\sqcup \{(m,k) \mid (0,k) \in Q^\psi_j\} \qquad \text{for any $j \ne i$;}
\\
Q^\phi_i &\ce \underQ_i^\psi \sqcup \{(0,m), (m,m)\} .
\end{align*}
\item[3] If $\phi = B_i (\xi\land \psi)$ for a DBI normal  formula $\psi$ and propositional formula $\xi$, the construction is the same as in the preceding case with the only change being that $\pre^\phi(m) \ce \xi$.
\item[4] If $\phi = \psi \land \theta$ for DBI normal formulas $\psi$ and $\theta$, then $\calU_\psi =\langle \{0,-1\}\sqcup D^\psi, Q^\psi, \pre^\psi\rangle$ and $\calU_\theta =\langle \{0,-1\}\sqcup D^\theta, Q^\theta, \pre^\theta\rangle$ have already been constructed.
We can assume w.l.o.g.~that $D^\psi \cap D^\theta = \varnothing$ (otherwise, just rename events in one of them). Take $D^\phi = D^\psi \sqcup D^\theta$, set $\pre^\phi = \pre^\psi \cup \pre^\theta$ and, for any $j$,
\begin{multline*}
Q^\phi_j \ce \underQ_j^\psi \cup \underQ_j^\theta \sqcup \{(0,k)\mid  (0,k) \in Q^\psi_j \cup Q^\theta_j, k \in D^\psi\sqcup D^\theta\} \sqcup\\ \{(0,-1)\mid (0,k)\notin Q^\psi_j \cup Q^\theta_j \text{ for any $k \in D^\psi\sqcup D^\theta$}\}.
\end{multline*}
\end{description}
 \end{definition}

We illustrate this update synthesis method by applying to Example~\ref{Ex:BLT}:

\begin{example}
 For the DBI normal form  $\varphi=B_b(B_t p \land B_l(p \land B_t p))$ of formula~\eqref{blt_goal} from Example~\ref{Ex:BLT},   pointed action model $(\calU_{\varphi}, 0)$ constructed   according to  \cref{def:synth_am1}  can be found in the left part of  Fig.~\ref{fig:blt_action_model} (stages of construction can be seen in Fig.~\ref{fig:modsyntree}).  The result of the pointed update of $(\calM,v)$ from the left part of Fig.~\ref{fig:LT_a} with this $(\calU_{\varphi}, 0)$ can be seen on the right of Fig.~\ref{fig:blt_action_model} (and is isomorphic to the right part of~\cref{fig:LT_a}). It is easy to see that $\calM\odot\calU_\varphi, (v,0) \vDash \phi$. Since all accessibility relations in this updated model~$\calM\odot\calU_\varphi$ are serial, it follows that it is common belief that all agents have consistent beliefs.

 In addition, we claim that this update synthesis has been achieved with minimal change to agents' beliefs. Indeed, it is easy to prove using $G$-bisimilarities for appropriate groups of agents that 
 \begin{align*}
 \calM,v&\equiv_{l,t}\calM\odot\calU_\varphi, (v,0)
 &
 \calM,v&\equiv_{b,l}\calM\odot\calU_\varphi, (u,1)
 \\
 \calM,v&\equiv_{b}\calM\odot\calU_\varphi, (u,2)
 &
 \calM,v&\equiv_{b,l}\calM\odot\calU_\varphi, (u,3)
 \end{align*}
  where  $v$ is the actual world and $u$ is the other world of~$\calM$
 (we omit set braces in the subscript of $\equiv$).
 The statements in the first column mean that  the update is imperceptible for agents $l$ and $t$ and that agent~$b$ believes that agent $l$ does not think $b$'s beliefs have changed. Similarly, the second column testifies that $b$ believes himself and thinks that $l$~believes that  $t$ does not think either $b$ or $l$ changed their beliefs. In other words, the only higher-order beliefs that are affected by the update are those explicitly dictated by~$\phi$. \looseness=-1
\end{example}
We illustrate the advantage of using pointed updates over the product update operation using Example~\ref{Ex:BLT}:
\begin{example}
Consider the same example but using the product update of standard DEL: The result of updating the initial model with the action model obtained from~$\varphi$ is shown in \cref{fig:BTLproduct1}.
In particular, the result of the product update has only one additional world w.r.t. the model obtained via pointed updates $\calM\odot\calU_\varphi$ of \cref{fig:blt_action_model} (left), which is not reachable from the actual world, namely the $p$ world in the $v$ cluster at the top.
However, the difference between the two operations is striking in case of iterated updates: applying to the resulting model the product update operation again with the  same action model leads to the Kripke model sketched in \cref{fig:blt_action_twice_update_model}, which contains a number of worlds that are not reachable from the actual world but copied in every sub-model nonetheless.
Compare this to the pointed update operation, whose result is the same for any number of iterations for the same action model (\cref{thm:synth_ac_idempotent}).
\end{example}

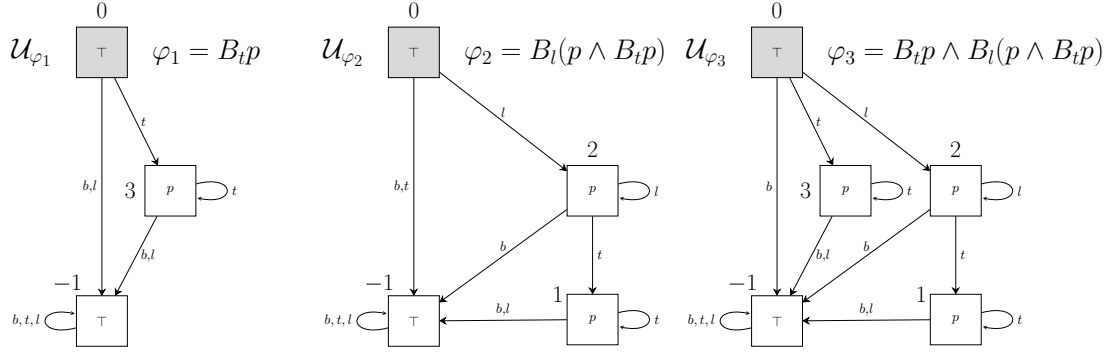
\begin{figure}[t]
    \centering
    \resizebox{15cm}{!}{%
\begin{subfigure}{0.6\textwidth}
\begin{tikzpicture}
    \begin{scope}[xshift=-6cm,>=stealth, yshift=3cm]
    \node[fill=gray!30,minimum size=1.3cm,draw] (0) {$\top$};
    \node [left=.5cm of 0] {\textbf{\huge{$\calU_{\varphi_1}$}}};
    \node [right=.5cm of 0] {\textbf{\huge{$\varphi_1= B_t p$}}};
    \node [above=.1cm of 0] {\Large{$0$}};
    \node (g) [below = 1cm  of 0] {};
    \node (n) [below = 1cm  of 0] {};
    \node (2) [right = 1.5cm  of n] {};
    \node[minimum size=1.3cm,draw] (3) [below = 1cm  of 2] {$p$};
    \node [left=.1cm of 3] {\Large{$3$}};
    \node[minimum size=1.3cm,draw] (6) [below = 5.6cm  of 0]  {$\top$};
    \node (x) [left =.1cm  of 6] {};
    \node [above=.5cm of x] {\Large{$-1$}};
    
            \path[->, -{Stealth[scale width=1.5]}]
        (0) edge node[right] {$t$} (3)
        (0) edge node[left] {$b$,$l$} (6)
        (3) edge node[right] {$b$,$l$} (6)
          ;
          \path[->,-{Stealth[scale width=1.5]}] 
          (3) edge[loop right] node[right] {$t$} (3)
          (6) edge[loop left] node[left] {$b,t,l$} (6);
    \end{scope}
\end{tikzpicture}
\label{8}
\end{subfigure}

\hfill

\begin{subfigure}{0.7\textwidth}
\begin{tikzpicture}
    \begin{scope}[xshift=-6cm,>=stealth, yshift=3cm]
    \node[fill=gray!30,minimum size=1.3cm,draw] (0) {$\top$};
    \node [left=.5cm of 0] {\textbf{\huge{$\calU_{\varphi_2}$}}};
    \node [right=.5cm of 0] {\textbf{\huge{$\varphi_2= B_l (p \wedge B_t p)$}}};
    \node [above=.1cm of 0] {\Large{$0$}};
    \node (g) [below = 1cm  of 0] {};
    \node (n) [below = 1cm  of 0] {};
    \node (2) [right = 1.5cm  of n] {};
    \node[minimum size=1.3cm,draw] (4) [right = 1.5cm  of 3] {$p$};
    \node [above=.1cm of 4] {\Large{$2$}};
    \node[minimum size=1.3cm,draw] (5) [below = 2cm  of 4] {$p$};
    \node (r) [left =.1cm  of 5] {};
    \node [above=.2cm of r] {\Large{$1$}};
    \node[minimum size=1.3cm,draw] (6) [below = 5.6cm  of 0]  {$\top$};
    \node (x) [left =.1cm  of 6] {};
    \node [above=.5cm of x] {\Large{$-1$}};
    
            \path[->, -{Stealth[scale width=1.5]}]
        (0) edge node[above] {$l$} (4)
        (0) edge node[left] {$b$,$t$} (6)
        (4) edge node[right] {$t$} (5)
        (4) edge node[above] {$b$} (6)
        (5) edge node[above] {$b$,$l$} (6)
          ;
          \path[->,-{Stealth[scale width=1.5]}] 
          (4) edge[loop right] node[right] {$l$} (4)
          (5) edge[loop right] node[right] {$t$} (5)
          (6) edge[loop left] node[left] {$b,t,l$} (6);
    \end{scope}
\end{tikzpicture}
\label{8}
\end{subfigure}

\hfill

\begin{subfigure}{0.9\textwidth}
\begin{tikzpicture}
    \begin{scope}[xshift=-6cm,>=stealth, yshift=3cm]
    \node[fill=gray!30,minimum size=1.3cm,draw] (0) {$\top$};
    \node [left=.5cm of 0] {\textbf{\huge{$\calU_{\varphi_3}$}}};
    \node [right=.5cm of 0] {\textbf{\huge{$\varphi_3= B_t p \wedge B_l (p \wedge B_t p)$}}};
    \node [above=.1cm of 0] {\Large{$0$}};
    \node (g) [below = 1cm  of 0] {};
    \node (n) [below = 1cm  of 0] {};
    \node (2) [right = 1.5cm  of n] {};
    \node[minimum size=1.3cm,draw] (3) [below = 1cm  of 2] {$p$};
    \node [left=.1cm of 3] {\Large{$3$}};
    \node[minimum size=1.3cm,draw] (4) [right = 1.5cm  of 3] {$p$};
    \node [above=.1cm of 4] {\Large{$2$}};
    \node[minimum size=1.3cm,draw] (5) [below = 2cm  of 4] {$p$};
    \node (r) [left =.1cm  of 5] {};
    \node [above=.2cm of r] {\Large{$1$}};
    \node[minimum size=1.3cm,draw] (6) [below = 5.6cm  of 0]  {$\top$};
    \node (x) [left =.1cm  of 6] {};
    \node [above=.5cm of x] {\Large{$-1$}};
    
            \path[->, -{Stealth[scale width=1.5]}]
        (0) edge node[right] {$t$} (3)
        (0) edge node[above] {$l$} (4)
        (0) edge node[left] {$b$} (6)
        (3) edge node[right] {$b$,$l$} (6)
        (4) edge node[right] {$t$} (5)
        (4) edge node[above] {$b$} (6)
        (5) edge node[above] {$b$,$l$} (6)
          ;
          \path[->,-{Stealth[scale width=1.5]}] 
          (3) edge[loop right] node[right] {$t$} (3)
          (4) edge[loop right] node[right] {$l$} (4)
          (5) edge[loop right] node[right] {$t$} (5)
          (6) edge[loop left] node[left] {$b,t,l$} (6);
    \end{scope}
\end{tikzpicture}
\label{8}
\end{subfigure}

    }
    \caption{Intermediate stages in constructing action model $\calU_\varphi$ for DBI formula~\eqref{blt_goal} from Example~\ref{Ex:BLT} based on its DBI normal form $\varphi=B_b(B_t p \land B_l(p \land B_t p))$. \looseness=-1}
    \label{fig:modsyntree}
\end{figure}
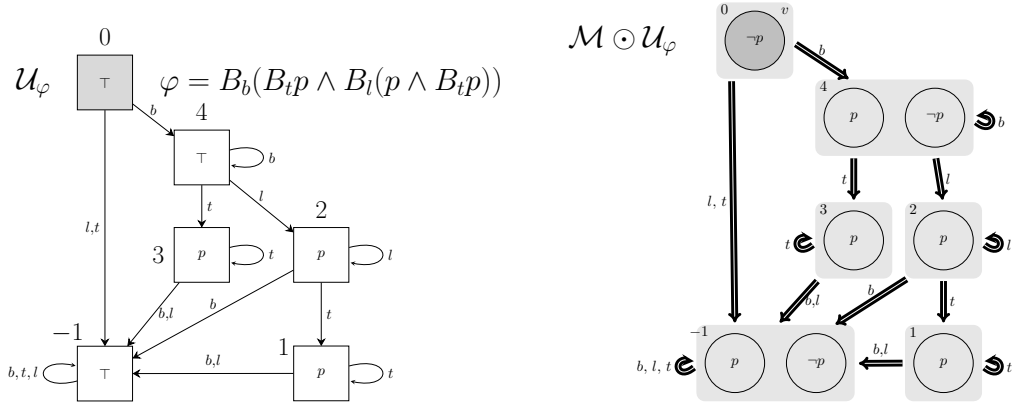
\begin{figure}[h!]
    \centering
    \resizebox{14cm}{!}{%

    \begin{subfigure}{1\textwidth}
\begin{tikzpicture}
    \begin{scope}[xshift=-6cm,>=stealth, yshift=3cm]
    \node[fill=gray!30,minimum size=1.3cm,draw] (0) {$\top$};
    \node [left=.5cm of 0] {\textbf{\huge{$\calU_\varphi$}}};
    \node [right=.5cm of 0] {\textbf{\huge{$\varphi= B_b (B_t p \wedge B_l (p \wedge B_t p))$}}};
    \node [above=.1cm of 0] {\Large{$0$}};
    \node (g) [below = 1cm  of 0] {};
    \node (n) [below = 1cm  of 0] {};
    \node[minimum size=1.3cm,draw] (2) [right = 1.5cm  of n] {$\top$};
    \node [above=.1cm of 2] {\Large{$4$}};
    \node[minimum size=1.3cm,draw] (3) [below = 1cm  of 2] {$p$};
    \node [left=.1cm of 3] {\Large{$3$}};
    \node[minimum size=1.3cm,draw] (4) [right = 1.5cm  of 3] {$p$};
    \node [above=.1cm of 4] {\Large{$2$}};
    \node[minimum size=1.3cm,draw] (5) [below = 1.5cm  of 4] {$p$};
    \node (r) [left =.1cm  of 5] {};
    \node [above=.2cm of r] {\Large{$1$}};
    \node[minimum size=1.3cm,draw] (6) [below = 5.6cm  of 0]  {$\top$};
    \node (x) [left =.1cm  of 6] {};
    \node [above=.5cm of x] {\Large{$-1$}};
    
            \path[->, -{Stealth[scale width=1.5]}]
        (0) edge node[above] {$b$} (2)
        (2) edge node[right] {$t$} (3)
        (2) edge node[above] {$l$} (4)
        (0) edge node[left] {$l$,$t$} (6)
        (3) edge node[right] {$b$,$l$} (6)
        (4) edge node[right] {$t$} (5)
        (4) edge node[above] {$b$} (6)
        (5) edge node[above] {$b$,$l$} (6)
          ;
          \path[->,-{Stealth[scale width=1.5]}] 
          (2) edge[loop right] node[right] {$b$} (2)
          (3) edge[loop right] node[right] {$t$} (3)
          (4) edge[loop right] node[right] {$l$} (4)
          (5) edge[loop right] node[right] {$t$} (5)
          (6) edge[loop left] node[left] {$b,t,l$} (6);
    \end{scope}
\end{tikzpicture}
\label{8}
\end{subfigure}

\hfill

\begin{subfigure}{0.9\textwidth}
\begin{tikzpicture}
    \begin{scope}[xshift=-6cm,>=stealth, yshift=3cm]
    \node[m,fill=gray!50,label=135:{0},label=45:{$v$}] (0) {$\neg p$};
    \node (g) [below = 1cm  of 0] {};
    \node [left=1cm of 0] {\textbf{\huge{$\calM\odot\calU_\varphi$}}};
    \node (n) [below = 1cm  of 0] {};
    \node[m,label=135:{4}] (2a) [right = 1.5cm  of n] {$p$};
    \node[m] (2b) [right = .5cm  of 2a] {$\neg p$};
    \node[m,,label=135:{3}] (3a) [below = 1.5cm  of 2a] {$p$};
    \node[m,label=135:{2}] (3b) [right = .7cm  of 3a] {$p$};
    \node[m,label=135:{1}] (4a) [below = 1.5cm  of 3b] {$p$};
    \node[m,label=135:{$-1$}] (6a) [left = 3.5cm  of 4a]  {$p$};
    \node[m] (6b) [right = .5cm  of 6a]  {$\neg p$};
    \node (g1) [right = 0.01cm  of 0] {};
    \node (g2) [above = 0.01cm  of 2a] {};
    \node (g3) [below = 0.01cm  of 2a] {};
    \node (g4) [above = 0.01cm  of 3a] {};
    \node (g5) [below = 0.01cm  of 2b] {};
    \node (g6) [above = 0.01cm  of 3b] {};
    \node (g7) [below = 0.01cm  of 3b] {};
    \node (g8) [above = 0.01cm  of 4a] {};
    \node (g9a) [below = 0.01cm  of 3a] {}; 
    \node (g9b) [left = 0.5cm  of g9a] {};
    \node (g10a) [above = 0.01cm  of 6a] {}; 
    \node (g10b) [right = .7cm  of g10a] {};
    \node (g11a) [below = 0.01cm  of 3b] {}; 
    \node (g11b) [left = 0.5cm  of g11a] {};
    \node (g12) [right = 2cm  of g10a] {};
    \node (g13) [left = 0.01cm  of 4a] {}; 
    \node (g14a) [right = 0.3cm  of g12] {}; 
    \node (g14b) [below = 0.6cm  of g14a] {};
    \node (g15a) [below = 0.01cm  of 0] {};
    \node (g15b) [left = 0.3cm  of g15a] {};
    \node (g16a) [right = 1cm  of g10a] {};
    \node (f1) [right = 0.01cm  of 2b] {};
    \node (f2) [right = 0.01cm  of 3b] {};
    \node (f3) [right = 0.01cm  of 4a] {};
    \node (f4) [left = 0.01cm  of 3a] {};
    \node (f5) [left = 0.01cm  of 6a] {};
    
    \draw[-{Implies},double,line width=2pt] (g1) -- (g2) node[midway,above] {$b$};
    \draw[-{Implies},double,line width=2pt] (g3) -- (g4) node[midway,left] {$t$};
    \draw[-{Implies},double,line width=2pt] (g5) -- (g6) node[midway,right] {$l$};
    \draw[-{Implies},double,line width=2pt] (g7) -- (g8) node[midway,right] {$t$};
    \draw[-{Implies},double,line width=2pt] (g9b) -- (g10b) node[midway,right] {$b$,$l$};
    \draw[-{Implies},double,line width=2pt] (g11b) -- (g12) node[midway,above] {$b$};
    \draw[-{Implies},double,line width=2pt] (g13) -- (g14b) node[midway,above] {$b$,$l$};
   \draw[-{Implies},double,line width=2pt] (g15b) -- (g10a) node[midway,left] {$l$, $t$};

    \path[->] (f1) edge[-{Implies},double,line width=2pt,reflexive right] node[right] {$b$} (f1);
    \path[->] (f2) edge[-{Implies},double,line width=2pt,reflexive right] node[right] {$l$} (f2);
    \path[->] (f3) edge[-{Implies},double,line width=2pt,reflexive right] node[right] {$t$} (f3);
    \path[->] (f4) edge[-{Implies},double,line width=2pt,reflexive left] node[left] {$t$} (f4);
    \path[->] (f5) edge[-{Implies},double,line width=2pt,reflexive left] node[left] {$b$, $l$, $t$} (f5);

    \end{scope}
\begin{pgfonlayer}{background}
\filldraw [line width=4mm,line join=round,black!10]
      (0.north  -| 0.east)  rectangle (0.south  -| 0.west);
 \filldraw [line width=4mm,line join=round,black!10]
      (2a.north  -| 2b.east)  rectangle (2a.south  -| 2a.west);
    \filldraw [line width=4mm,line join=round,black!10]
      (3a.north  -| 3a.east)  rectangle (3a.south  -| 3a.west);
       \filldraw [line width=4mm,line join=round,black!10]
      (3b.north  -| 3b.east)  rectangle (3b.south  -| 3b.west);
       \filldraw [line width=4mm,line join=round,black!10]
      (4a.north  -| 4a.east)  rectangle (4a.south  -| 4a.west);
      \filldraw [line width=4mm,line join=round,black!10]
      (6a.north  -| 6b.east)  rectangle (6a.south  -| 6a.west);
  \end{pgfonlayer}
\end{tikzpicture}
\label{5}
\end{subfigure}

    }
    \caption{Successful update synthesis $\calM\odot\calU_\varphi, (v,0) \vDash \phi$ by applying pointed update  with pointed action model $(\calU_\phi,0)$ for DBI formula~\eqref{blt_goal} from Example~\ref{Ex:BLT} to~$(\calM,v)$ from Fig.~\ref{fig:LT_a} (left).
    (Legend: the real event and world are dark gray; light gray rectangles are a compact way for representing arrows in the updated model: a double arrow with label $a\in\agents$ from one rectangle to another represents a set of accessibility arrows with label~$a$ from every world in the source rectangle to every world in the target rectangle).
    \looseness=-1}
    \label{fig:blt_action_model}
\end{figure}

\begin{figure}
    \centering
    \resizebox{7cm}{!}{%
\begin{tikzpicture}
    \begin{scope}[xshift=-6cm,>=stealth, yshift=3cm]
    \node[m,fill=gray!40,label=45:{$v$}] (0) {$\neg p$};
    \node[m] (0a) [left =.5cm  of 0] {$p$};
    \node (g) [below = 1cm  of 0] {};
    \node [left=1.3cm of 0a] {\textbf{\huge{$\calM \otimes \calU_\varphi$}}};
    \node (n) [below = 1cm  of 0] {};
    \node[m] (2a) [right = 1.5cm  of n] {$p$};
    \node[m] (2b) [right = .5cm  of 2a] {$\neg p$};
    \node[m] (3a) [below = 1.5cm  of 2a] {$p$};
    \node[m] (3b) [right = .7cm  of 3a] {$p$};
    \node[m] (4a) [below = 1.5cm  of 3b] {$p$};
    \node[m] (6a) [left = 3.5cm  of 4a]  {$p$};
    \node[m] (6b) [right = .5cm  of 6a]  {$\neg p$};
    \node (g1) [right = 0.01cm  of 0] {};
    \node (g2) [above = 0.01cm  of 2a] {};
    \node (g3) [below = 0.01cm  of 2a] {};
    \node (g4) [above = 0.01cm  of 3a] {};
    \node (g5) [below = 0.01cm  of 2b] {};
    \node (g6) [above = 0.01cm  of 3b] {};
    \node (g7) [below = 0.01cm  of 3b] {};
    \node (g8) [above = 0.01cm  of 4a] {};
    \node (g9a) [below = 0.01cm  of 3a] {}; 
    \node (g9b) [left = 0.5cm  of g9a] {};
    \node (g10a) [above = 0.01cm  of 6a] {}; 
    \node (g10b) [right = .7cm  of g10a] {};
    \node (g11a) [below = 0.01cm  of 3b] {}; 
    \node (g11b) [left = 0.5cm  of g11a] {};
    \node (g12) [right = 2cm  of g10a] {};
    \node (g13) [left = 0.01cm  of 4a] {}; 
    \node (g14a) [right = 0.3cm  of g12] {}; 
    \node (g14b) [below = 0.6cm  of g14a] {};
    \node (g15a) [below = 0.01cm  of 0] {};
    \node (g15b) [left = 0.3cm  of g15a] {};
    \node (g16a) [right = 1cm  of g10a] {};
    \node (f1) [right = 0.01cm  of 2b] {};
    \node (f2) [right = 0.01cm  of 3b] {};
    \node (f3) [right = 0.01cm  of 4a] {};
    \node (f4) [left = 0.01cm  of 3a] {};
    \node (f5) [left = 0.01cm  of 6a] {};

    \draw[-{Implies},double,line width=2pt] (g1) -- (g2) node[midway,above] {$b$};
    \draw[-{Implies},double,line width=2pt] (g3) -- (g4) node[midway,left] {$t$};
    \draw[-{Implies},double,line width=2pt] (g5) -- (g6) node[midway,right] {$l$};
    \draw[-{Implies},double,line width=2pt] (g7) -- (g8) node[midway,right] {$t$};
    \draw[-{Implies},double,line width=2pt] (g9b) -- (g10b) node[midway,right] {$b$,$l$};
    \draw[-{Implies},double,line width=2pt] (g11b) -- (g12) node[midway,above] {$b$};
    \draw[-{Implies},double,line width=2pt] (g13) -- (g14b) node[midway,above] {$b$,$l$};
   \draw[-{Implies},double,line width=2pt] (g15b) -- (g10a) node[midway,right] {$l$, $t$};
    
    \path[->] (f1) edge[-{Implies},double,line width=2pt,reflexive right] node[right] {$b$} (f1);
    \path[->] (f2) edge[-{Implies},double,line width=2pt,reflexive right] node[right] {$l$} (f2);
    \path[->] (f3) edge[-{Implies},double,line width=2pt,reflexive right] node[right] {$t$} (f3);
    \path[->] (f4) edge[-{Implies},double,line width=2pt,reflexive left] node[left] {$t$} (f4);
    \path[->] (f5) edge[-{Implies},double,line width=2pt,reflexive left] node[left] {$b$, $l$, $t$} (f5);

    \end{scope}
\begin{pgfonlayer}{background}
\filldraw [line width=4mm,line join=round,black!10]
      (0.north  -| 0.east)  rectangle (0.south  -| 0a.west);
 \filldraw [line width=4mm,line join=round,black!10]
      (2a.north  -| 2b.east)  rectangle (2a.south  -| 2a.west);
    \filldraw [line width=4mm,line join=round,black!10]
      (3a.north  -| 3a.east)  rectangle (3a.south  -| 3a.west);
       \filldraw [line width=4mm,line join=round,black!10]
      (3b.north  -| 3b.east)  rectangle (3b.south  -| 3b.west);
       \filldraw [line width=4mm,line join=round,black!10]
      (4a.north  -| 4a.east)  rectangle (4a.south  -| 4a.west);
      \filldraw [line width=4mm,line join=round,black!10]
      (6a.north  -| 6b.east)  rectangle (6a.south  -| 6a.west);
  \end{pgfonlayer}
\end{tikzpicture}
}
\label{9}
\caption{Balder--Loki--Thor updated model using product updates  $\calM \otimes \mathcal{U}_\varphi$. The real world $v$ is underlined. 
    Grey rectangles are a compact way for representing arrows: A double arrow from a rectangle to another represents a set of accessibility arrows with the same agent's label from every world in the source rectangle to every world in the target rectangle.\looseness=-1}
    \label{fig:BTLproduct1}
\end{figure}
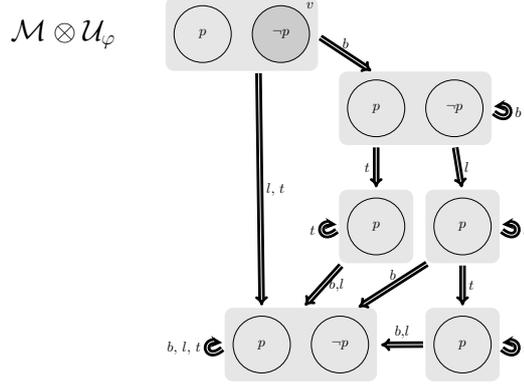

\begin{figure}[h!]
    \centering
     \resizebox{9cm}{!}{%
\begin{tikzpicture}
    \begin{scope}[xshift=-6cm,>=stealth, yshift=3cm]
    \node[m,fill=gray!40,label=45:{$v$}] (0) {$\neg p$};
    \node[m] (0a) [left =.5cm  of 0] {$p$};
    \node (g) [below = 1cm  of 0] {};
    \node [left=9.5cm of 3b] {\textbf{\Huge{$(\calM \otimes \calU_\varphi) \otimes \calU_\varphi$}}};
    \node (n) [below = 1cm  of 0] {};
    \node[m] (2a) [right = 1.5cm  of n] {$p$};
    \node[m] (2b) [right = .5cm  of 2a] {$\neg p$};
    \node[m] (3a) [below = 1.5cm  of 2a] {$p$};
    \node[m] (3b) [right = .7cm  of 3a] {$p$};
    \node[m] (4a) [below = 1.5cm  of 3b] {$p$};
    \node[m] (6a) [left = 3.5cm  of 4a]  {$p$};
    \node[m] (6b) [right = .5cm  of 6a]  {$\neg p$};
    
    \node (r) [below = 3cm  of 4a] {};

    \node[m] (0r) [right =3.5cm  of r]{$\neg p$};
    \node[m] (0ar) [left =.5cm  of 0r] {$p$};
    \node (gr) [below = 1cm  of 0r] {};
    \node (nr) [below = 1cm  of 0r] {};
    \node[m] (2ar) [right = 1.5cm  of nr] {$p$};
    \node[m] (2br) [right = .5cm  of 2ar] {$\neg p$};
    \node[m] (3ar) [below = 1.5cm  of 2ar] {$p$};
    \node[m] (3br) [right = .7cm  of 3ar] {$p$};
    \node[m] (4ar) [below = 1.5cm  of 3br] {$p$};
    \node[m] (6ar) [left = 3.5cm  of 4ar]  {$p$};
    \node[m] (6br) [right = .5cm  of 6ar]  {$\neg p$};

    \node (q) [below = 3cm  of 6br] {};

    \node (0q) [below =.5cm  of q]{};
    \node[m] (0aq) [left =.5cm  of 0q] {$p$};
    \node (gq) [below = 1cm  of 0q] {};
    \node (nq) [below = 1cm  of 0q] {};
    \node[m] (2aq) [right = 1.5cm  of nq] {$p$};
    \node (2bq) [right = .5cm  of 2aq] {};
    \node[m] (3aq) [below = 1.5cm  of 2aq] {$p$};
    \node[m] (3bq) [right = .7cm  of 3aq] {$p$};
    \node[m] (4aq) [below = 1.5cm  of 3bq] {$p$};
    \node[m] (6aq) [left = 3.5cm  of 4aq]  {$p$};
    \node (6bq) [right = .5cm  of 6aq] {};

    \node (s) [left = 12cm  of 0q] {};

    \node (0s) [left =.5cm  of s]{};
    \node[m] (0as) [left =.5cm  of 0s] {$p$};
    \node (gs) [below = 1cm  of 0s] {};
    \node (ns) [below = 1cm  of 0s] {};
    \node[m] (2as) [right = 1.5cm  of ns] {$p$};
    \node (2bs) [right = .5cm  of 2as] {};
    \node[m] (3as) [below = 1.5cm  of 2as] {$p$};
    \node[m] (3bs) [right = .7cm  of 3as] {$p$};
    \node[m] (4as) [below = 1.5cm  of 3bs] {$p$};
    \node[m] (6as) [left = 3.5cm  of 4as]  {$p$};
    \node (6bs) [right = .5cm  of 6as]  {};

    \node (p) [below = 3cm  of 6bq] {};

    \node (0p) [below =.5cm  of p]{};
    \node[m] (0ap) [left =.5cm  of 0p] {$p$};
    \node (gp) [below = 1cm  of 0p] {};
    \node (np) [below = 1cm  of 0p] {};
    \node[m] (2ap) [right = 1.5cm  of np] {$p$};
    \node (2bp) [right = .5cm  of 2ap] {};
    \node[m] (3ap) [below = 1.5cm  of 2ap] {$p$};
    \node[m] (3bp) [right = .7cm  of 3ap] {$p$};
    \node[m] (4ap) [below = 1.5cm  of 3bp] {$p$};
    \node[m] (6ap) [left = 3.5cm  of 4ap]  {$p$};
    \node (6bp) [right = .5cm  of 6ap]  {};

    \node (f) [left = 24cm  of 0p] {};

    \node[m] (0f) [left =.5cm  of f]{$\neg p$};
    \node[m] (0af) [left =.5cm  of 0f] {$p$};
    \node (gf) [below = 1cm  of 0f] {};
    \node (nf) [below = 1cm  of 0f] {};
    \node[m] (2af) [right = 1.5cm  of nf] {$p$};
    \node[m] (2bf) [right = .5cm  of 2af] {$\neg p$};
    \node[m] (3af) [below = 1.5cm  of 2af] {$p$};
    \node[m] (3bf) [right = .7cm  of 3af] {$p$};
    \node[m] (4af) [below = 1.5cm  of 3bf] {$p$};
    \node[m] (6af) [left = 3.5cm  of 4af]  {$p$};
    \node[m] (6bf) [right = .5cm  of 6af]  {$\neg p$};
    
    \node (g1) [right = 0.01cm  of 0f] {};
    \node (g2) [above = 0.01cm  of 2af] {};
    \node (g3) [below = 0.01cm  of 2af] {};
    \node (g4) [above = 0.01cm  of 3af] {};
    \node (g5) [below = 0.01cm  of 2bf] {};
    \node (g6) [above = 0.01cm  of 3bf] {};
    \node (g7) [below = 0.01cm  of 3bf] {};
    \node (g8) [above = 0.01cm  of 4af] {};
    \node (g9a) [below = 0.01cm  of 3af] {}; 
    \node (g9b) [left = 0.5cm  of g9a] {};
    \node (g10a) [above = 0.01cm  of 6af] {}; 
    \node (g10b) [right = .7cm  of g10a] {};
    \node (g11a) [below = 0.01cm  of 3bf] {}; 
    \node (g11b) [left = 0.5cm  of g11a] {};
    \node (g12) [right = 2cm  of g10a] {};
    \node (g13) [left = 0.01cm  of 4af] {}; 
    \node (g14a) [right = 0.3cm  of g12] {}; 
    \node (g14b) [below = 0.6cm  of g14a] {};
    \node (g15a) [below = 0.01cm  of 0f] {};
    \node (g15b) [left = 0.3cm  of g15a] {};
    \node (g16a) [right = 1cm  of g10a] {};
    \node (f1) [right = 0.01cm  of 2bf] {};
    \node (f2) [right = 0.01cm  of 3bf] {};
    \node (f3) [right = 0.01cm  of 4af] {};
    \node (f4) [left = 0.01cm  of 3af] {};
    \node (f5) [left = 0.01cm  of 6af] {};

\node (D0) [below = 0.01cm  of 6b] {};
\node (D1) [above = 0.01cm  of 0ar] {};
\draw[->,dotted,line width=2pt] (D0) -- (D1) node[midway,above] {$b$};
\node (D2) [below = 0.01cm  of 6a] {};
\node (D3) [above = 0.01cm  of 0af] {};
\draw[->,dotted,line width=2pt] (D2) -- (D3) node[midway,left] {$l,t$};
\node (D4) [below = 0.01cm  of 6ar] {};
\node (D5) [above = 0.01cm  of 0aq] {};
\draw[->,dotted,line width=2pt] (D4) -- (D5) node[midway,right] {$l$};
\node (D6) [left = 0.1cm  of D4] {};
\node (D7) [above = 4.1cm  of 3bs] {};
\draw[->,dotted,line width=2pt] (D6) -- (D7) node[midway,below] {$t$};
\node (D8) [below = 0.01cm  of 6aq] {};
\node (D9) [above = 0.01cm  of 0ap] {};
\draw[->,dotted,line width=2pt] (D8) -- (D9) node[midway,right] {$l$};
\node (D10) [left = 0.1cm  of D8] {};
\node (D11) [right =1cm  of 3bf] {};
\draw[->,dotted,line width=2pt] (D10) -- (D11) node[midway,right] {$b$};
\node (D12) [below = 0.1cm  of 6as] {};
\node (D13) [right =1cm  of 2bf] {};
\draw[->,dotted,line width=2pt] (D12) -- (D13) node[midway,right] {$b,l$};
\node (D14) [left = 6.1cm  of 3bp] {};
\node (D15) [below =.1cm  of D11] {};
\draw[->,dotted,line width=2pt] (D14) -- (D15) node[midway,above] {$b,l$};

\node (L1) [right=.1cm  of 3br] {};
\node (L2) [right=.1cm  of 2br] {};
\path[->] (L1) edge[->, bend right=85, dotted,line width=2pt] node[right] {$b$} (L2);
\node (L3) [right=.1cm  of 3bs] {};
\node (L4) [right=1.4cm  of 2bs] {};
\path[->] (L3) edge[->, bend right=85, dotted,line width=2pt] node[right] {$t$} (L4);
\node (L5) [right=.1cm  of 3bq] {};
\node (L6) [right=1.4cm  of 2bq] {};
\path[->] (L5) edge[->, bend right=85, dotted,line width=2pt] node[right] {$l$} (L6);
\node (L7) [right=.1cm  of 3bp] {};
\node (L8) [right=1.4cm  of 2bp] {};
\path[->] (L7) edge[->, bend right=85, dotted,line width=2pt] node[right] {$t$} (L8);

    \draw[-{Implies},double,line width=2pt] (g1) -- (g2) node[midway,above] {$b$};
    \draw[-{Implies},double,line width=2pt] (g3) -- (g4) node[midway,left] {$t$};
    \draw[-{Implies},double,line width=2pt] (g5) -- (g6) node[midway,right] {$l$};
    \draw[-{Implies},double,line width=2pt] (g7) -- (g8) node[midway,right] {$t$};
    \draw[-{Implies},double,line width=2pt] (g9b) -- (g10b) node[midway,right] {$b$,$l$};
    \draw[-{Implies},double,line width=2pt] (g11b) -- (g12) node[midway,above] {$b$};
    \draw[-{Implies},double,line width=2pt] (g13) -- (g14b) node[midway,above] {$b$,$l$};
   \draw[-{Implies},double,line width=2pt] (g15b) -- (g10a) node[midway,right] {$l$, $t$};
    
    \path[->] (f1) edge[-{Implies},double,line width=2pt,reflexive right] node[right] {$b$} (f1);
    \path[->] (f2) edge[-{Implies},double,line width=2pt,reflexive right] node[right] {$l$} (f2);
    \path[->] (f3) edge[-{Implies},double,line width=2pt,reflexive right] node[right] {$t$} (f3);
    \path[->] (f4) edge[-{Implies},double,line width=2pt,reflexive left] node[left] {$t$} (f4);
    \path[->] (f5) edge[-{Implies},double,line width=2pt,reflexive left] node[left] {$b$, $l$, $t$} (f5);

    \end{scope}
\begin{pgfonlayer}{background}
\filldraw [line width=4mm,line join=round,black!10]
      (0f.north  -| 0f.east)  rectangle (0f.south  -| 0af.west);
 \filldraw [line width=4mm,line join=round,black!10]
      (2af.north  -| 2bf.east)  rectangle (2af.south  -| 2af.west);
    \filldraw [line width=4mm,line join=round,black!10]
      (3af.north  -| 3af.east)  rectangle (3af.south  -| 3af.west);
       \filldraw [line width=4mm,line join=round,black!10]
      (3bf.north  -| 3bf.east)  rectangle (3bf.south  -| 3bf.west);
       \filldraw [line width=4mm,line join=round,black!10]
      (4af.north  -| 4af.east)  rectangle (4af.south  -| 4af.west);
      \filldraw [line width=4mm,line join=round,black!10]
      (6af.north  -| 6bf.east)  rectangle (6af.south  -| 6af.west);

      \filldraw [line width=4mm,line join=round,black!10]
      (0a.north  -| 2b.east)  rectangle (6a.south  -| 0a.west);
      \filldraw [line width=4mm,line join=round,black!10]
      (0ar.north  -| 2br.east)  rectangle (6ar.south  -| 0ar.west);
      \filldraw [line width=4mm,line join=round,black!10]
      (0aq.north  -| 3bq.east)  rectangle (6aq.south  -| 0aq.west);
      \filldraw [line width=4mm,line join=round,black!10]
      (0as.north  -| 3bs.east)  rectangle (6as.south  -| 0as.west);
      \filldraw [line width=4mm,line join=round,black!10]
      (0ap.north  -| 3bp.east)  rectangle (6ap.south  -| 0ap.west);
  \end{pgfonlayer}
\end{tikzpicture}
}
	\caption{Sketch of the result of the iterated product update operation $(\calM \otimes \calU_\varphi) \otimes \calU_\varphi $.
	Grey rectangles are a compact way to represent sub-models originating from worlds in the previous model.
Dotted arrows are meant to capture the direction of the arrows and the agents involved.
Unlike double arrows, they do not represent arrows from and to all worlds in the rectangle, but only from and to some of these worlds.
We leave to the reader their identification.
Finally, note that the sink is a copy of the previous model.}
    \label{fig:blt_action_twice_update_model}
\end{figure}
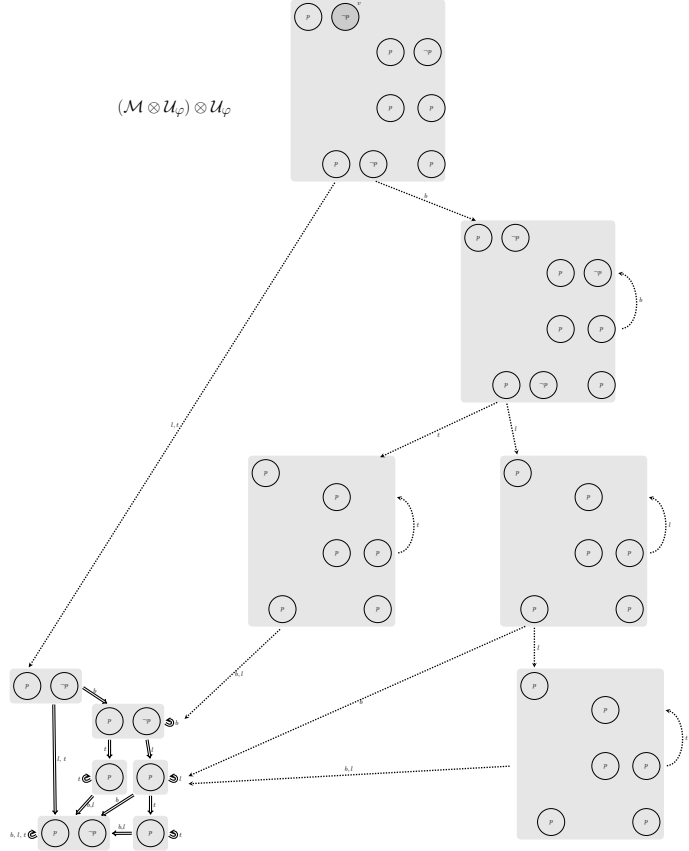

Let us now prove that $\calU_\phi$ performs update synthesis for any DBI normal $\phi$, minimizes belief change of other agents, and preserves consistency whenever possible. 

We will  use the auxiliary fact that updates do not affect propositional formulas:
\begin{lemma}[Propositional invariance]
\label{lem:prop_upd_stay}
Let $(w,\alpha)$ belong to the domain of $\calM\odot\calU$ for some Kripke model $\calM$ and action model $\calU$. Then for any  propositional formula~$\xi$,  
\begin{equation}
\label{eq:prop_stay_same}
    \calM,w \vDash \xi \qquad\Longleftrightarrow\qquad \calM\odot\calU, (w,\alpha) \vDash \xi.
\end{equation}
\end{lemma}
\begin{proof}
For product updates $\otimes$, this is well known (see, e.g.,~\cite{ditmarsch2007dynamic}). For pointed updates, it follows from Theorem~\ref{th:sameupdates}.
\end{proof}

\begin{theorem}[Update synthesis success] \label{thm:update_synth_succ}
For any DBI normal formula $\phi$ and any pointed Kripke model $(\calM,v)$, the pointed update of $(\calM,v)$ with $(\calU_\phi,0)$ is defined and 
\begin{equation}
\label{eq:synthesis}
\calM\odot\calU_\phi, (v,0) \vDash \phi.
\end{equation}
\end{theorem}
\begin{proof}
Let $\calM=\langle S, R, V\rangle$ with $v \in S$. Per Def.~\ref{def:synth_am1}, $\calU_\phi = \langle E^\phi,Q^\phi, \pre^\phi\rangle$.
The pointed update $\calM\odot\calU_\phi, (v,0)$ is defined because $\pre^\phi(0)=\top$ always holds.  Let \mbox{$\calM\odot\calU_\phi = \langle W^\phi,R^\phi,V^\phi\rangle$}. We prove~\eqref{eq:synthesis} by induction on the construction of $\phi$. For $\phi = B_i \xi$ with propositional $\xi$, whenever $(v,0)R_i^\phi(u,m)$, then $\calM,u\vDash \xi$ because $\pre^\phi(m) =\xi$. Hence, $\calM\odot\calU_\phi, (u,m) \vDash \xi$ by Lemma~\ref{lem:prop_upd_stay}. We conclude \eqref{eq:synthesis} for $\phi= B_i \xi$ since $(u,m)$ was chosen arbitrarily. 

For $\phi = B_i \psi$, it follows by construction that $\calU_\phi,m \leftrightarroweq_{\agents\setminus\{i\}} \calU_\psi,0$.   Whenever $v R_i u$,
since bisimilarity is reflexive, $\calM\odot\calU_\phi, (u,m) \leftrightarroweq_{\agents \setminus \{i\}} \calM\odot\calU_\psi, (u,0)$ by Theorem~\ref{th:sameupdates}. By IH, $\calM \odot \calU_\psi, (u,0) \vDash \psi$. By normality of $\phi$, $i \notin\ta(\psi)$. Hence, $\calM \odot \calU_\phi, (u,m) \vDash \psi$  by Lemma~\ref{lem:Gbisimgoal}. Now \eqref{eq:synthesis} for $\phi=B_i \psi$ follows since  $(v,0)R_i^\phi(u,\alpha)$ if{f} $vR_iu$ and $\alpha = m$ due to $\pre^\phi(m) = \top$. The case of $\phi = B_i(\xi \land \psi)$ is analogous, except additionally $\calM\odot\calU_\phi, (u,m) \vDash \xi$ because $\pre^\phi(m)=\xi$, as in the case of $B_i\xi$.

For $\phi = \psi_1\land \psi_2$, by construction and normality of $\phi$ we have $\calU_\phi,0 \leftrightarroweq_{\ta(\psi_j)} \calU_{\psi_j},0$ and, hence, $\calM\odot\calU_\phi,(v,0) \leftrightarroweq_{\ta(\psi_j)} \calM\odot\calU_{\psi_j},(v,0)$   for $j=1,2$.
$\calM \odot \calU_{\psi_j}, (v,0) \vDash \psi_j$ for $j=1,2$ by IH. Hence, $\calM \odot \calU_\phi, (v,0) \vDash \psi_j$ for $j=1,2$ by Lemma~\ref{lem:Gbisimgoal}. Now \eqref{eq:synthesis} for $\phi=\psi \land \theta$ follows immediately. 
\end{proof}

To formulate statements about minimal change, we need the following definition:
\begin{definition}[Independent formulas]
Let  $\phi$  be a DBI normal formula  and its corresponding action model $\calU_\phi=\langle E^\phi, Q^\phi,\pre^\phi\rangle$  be constructed according to Def.~\ref{def:synth_am1}. A modal formula $\theta$ is in \emph{$\top$-shape w.r.t.~event $\alpha_0 \in E^\phi$} if{f} for each sequence of modal operators $B_{i_1}\dots B_{i_k}$ used in the construction of $\theta$, including the empty sequence $\varepsilon$ with $k=0$, and for the unique corresponding sequence  $\alpha_0 Q^\phi_{i_1}\alpha_1Q^\phi_{i_2}\alpha_2\dots\alpha_{k-1}Q^\phi_{i_k}\alpha_k$ of events $\alpha_j \in E^\phi$ all $\pre^\phi(\alpha_j)=\top$ for  $0\leq j\leq k$. Formula $\theta$ is 
called \emph{independent} of~$\phi$ if{f} it is in $\top$-shape w.r.t.~$0\in E^\phi$. If $\theta$ is not independent of $\phi$, it is \emph{dependent} on $\phi$.
\end{definition}

\begin{example}
For formula  $\phi=B_b\bigl(B_t p \land B_l(p \land B_tp)\bigr)$ from Example~\ref{Ex:BLT}, we have that $\theta =  \lnot B_b B_b p \lor B_t B_l  p$ is independent from $\phi$ because the modality sequences 
\[
\varepsilon, B_b, B_bB_b, B_t, \text{ and }B_tB_l
\]  
correspond to event sequences 
\[
0, 0Q^\phi_b4, 0Q^\phi_b4Q^\phi_b4, 0Q^\phi_t\sink, \text{ and } 0Q^\phi_t\sink Q^\phi_l\sink
\] 
and preconditions for events $0$, $4$, and $\sink$ are all $\top$.  On the other hand, \mbox{$\eta =  \lnot B_b B_t p$} is dependent on $\phi$ because modality sequence $B_bB_t$ corresponds to event sequence~$0Q^\phi_b4Q^\phi_t3$, and $\pre^\phi(3) =p\ne \top$. Both $\theta$ and $\eta$ were true in $\calM,v$ from Fig.~\ref{fig:LT_a}, left. While the pointed update with $(\calU_\phi,0)$ keeps $\theta$ true, formula $\eta$ becomes false  in $\calM \odot \calU_\phi, (v,0)$ (see Fig.~\ref{fig:blt_action_model}, right).
\end{example}

\begin{theorem}[Update synthesis minimality]
\label{thm:minimalityr}
If formula $\theta$ is independent of a DBI~normal formula $\phi$, then for any pointed Kripke model $(\calM,v)$,
\begin{equation}
\label{eq:goalind}
\calM,v \vDash \theta
\qquad \Longleftrightarrow \qquad
\calM\odot\calU_\phi, (v,0) \vDash \theta.
\end{equation}
\end{theorem}

\begin{proof}
Let $\calM=\langle S, R, V\rangle$.
We prove by induction on the construction of $\theta$ that, for any $\alpha \in E^\phi$ and any $u \in S$, if $\theta$  is in $\top$-shape w.r.t.~$\alpha$, then $\calM, u \vDash \theta$ if{f} $\calM\odot\calU_\phi, (u,\alpha) \vDash \theta$. \eqref{eq:goalind} is an instance of this induction statement for $\alpha=0$ and $u=v$. Note that if $\theta$~is in $\top$-shape w.r.t.~$\alpha$, then $\pre^\phi(\alpha) = \top$ because of the empty sequence~$\varepsilon$, hence, the pointed update of $(\calM,u)$ with $(\calU_\phi,\alpha)$ is defined for all $u \in S$. 
For propositional atoms, the statement follows from the definition of pointed updates. The cases for Boolean connectives are straightforward. It remains to show  the induction statement for  $\theta = B_i \eta$. Let $\beta \in E^\phi$ be the unique event such that $\alpha Q_i^\phi\beta$. Then $\pre^\phi(\beta) = \top$ because of sequence $B_i$ of $\theta$ and, additionally, $\eta$ is in $\top$-shape w.r.t.~$\beta$. By IH, for any $u \in S$, we have $\calM,u\vDash \eta$ if{f} $\calM\odot\calU_\phi, (u,\beta) \vDash \eta$. We have 
$\calM, u \nvDash B_i \eta$ if{f} $\calM, w \nvDash\eta$ for some $uR_iw$. By IH, this is equivalent to $\calM\odot\calU_\phi, (w,\beta) \nvDash\eta$, which is equivalent to $\calM\odot\calU_\phi, (u,\alpha) \nvDash B_i\eta$ because $(u,\alpha)R^\phi_i(w,\gamma)$ if{f} $u R_i w$ and $\gamma=\beta$.
\end{proof}

\begin{corollary}
\label{cor:minim}
Let $\phi = \bigwedge_{j\in G}B_j\omega_j$ be a DBI normal formula, where $\varnothing \ne G \subseteq \agents$.
\begin{compactenum}
\item
If $i \notin \ta(\phi)$, i.e., if $ i\notin G$, then $i$'s beliefs are unaffected by pointed update $\calU_\phi$, i.e., $\calM,v \vDash B_i
\sigma$ if{f}
$\calM\odot\calU_\phi, (v,0) \vDash B_i \sigma$ for any formula $\sigma$ and any pointed Kripke model $(\calM,v)$.
\item If $i\in \ta(\phi)$, but $\omega_i = \bigwedge_{j\in H}B_j\pi_j$ has no propositional component, then $i$'s propositional beliefs are unaffected by $\calU_\phi$, i.e., $\calM,v \vDash B_i
\chi$ if{f}
$\calM\odot\calU_\phi, (v,0) \vDash B_i \chi$ for any propositional formula $\chi$ and any pointed Kripke model $(\calM,v)$.
\end{compactenum}
\end{corollary}

\begin{theorem}[Update synthesis consistency preservation]
\label{thm:minimality}
Let $\phi$ be a DBI normal formula.
Pointed update with $\calU_\phi$ can only cause (higher-order) inconsistent beliefs if the respective (higher-order) beliefs originally excluded the respective propositional preconditions: for any modality sequence  $B_{i_1}\dots B_{i_k}$ with $i_j\ne i_{j+1}$ for any $j$ and any pointed Kripke model $(\calM,v)$,
\begin{multline}
\label{eq:conserve}
\calM,v \vDash \hB_{i_1}\Bigl(\pre^\phi(\alpha_1) \land \hB_{i_2}\left(\pre^\phi(\alpha_2) \land \dots \hB_{i_k} \pre^\phi(\alpha_k)\right)\Bigr)
\qquad \Longleftrightarrow \qquad \\
\calM\odot\calU_\phi, (v,0) \nvDash B_{i_1}\dots B_{i_k} \bot
\end{multline}
for the unique sequence $0 Q^\phi_{i_1}\alpha_1Q^\phi_{i_2}\alpha_2\dots Q^\phi_{i_k}\alpha_k$ of events $\alpha_j \in E^\phi$.
\end{theorem}
\begin{proof}
Let $\calM= \langle S, R, V\rangle$ and $\calM \odot \calU_\phi = \langle S^\phi, E^\phi, V^\phi\rangle$ for the pointed update of $(\calM,v)$ with $(\calU_\phi,0)$.
The left statement holds if{f}  there is a sequence  $vR_{i_1}s_1R_{i_2}s_2\dots R_{i_k}s_k$ of states from $S$ such that $\calM, s_j \vDash \pre^\phi(\alpha_j)$ for $j=1,\dots,k$. It is easy to observe (by induction on $k$) that this is equivalent to  $(s_j,\alpha_j) \in E^\phi$ for $j=1,\dots,k$ and 
$
(v,0)R^\phi_{i_1}(s_1,\alpha_1)R^\phi_{i_2}(s_2,\alpha_2)\dots R^\phi_{i_k}(s_k,\alpha_k)
$. The equivalence to the right statement now follows from the fact that  $\calM\odot\calU_\phi, (s_k,\alpha_k) \nvDash \bot$ and the uniqueness of the sequence of $\alpha_j$'s such that $0 Q^\phi_{i_1}\alpha_1Q^\phi_{i_2}\alpha_2\dots Q^\phi_{i_k}\alpha_k$. 
\end{proof}

\begin{corollary}
Let $\phi$ be a DBI normal formula. 
\begin{compactenum}
    \item
    If $\phi$ fits either of the clauses of Cor.~\ref{cor:minim} for agent~$i$, then $\calU_\phi$ preserves $i$'s~consistency, i.e., $\calM,v \nvDash B_i \bot$ if{f} $\calM\odot\calU_\phi, (v,0) \nvDash B_{i} \bot$ for any $(\calM,v)$.
\item
    If $\phi = B_i\left(\xi \land \bigwedge_{k\in H} B_k \pi_k\right) \land \bigwedge_{j\in G} B_j \omega_j$ where $i \notin G$ and $\xi$ is propositional, then $\calU_\phi$ makes  $i$'s beliefs inconsistent if and only if $i$ originally does not consider $\xi$ to be possible, i.e.,  $\calM,v \vDash  \hB_i  \xi$ if{f} $\calM\odot\calU_\phi, (v,0) \nvDash B_{i} \bot$ for any $(\calM,v)$.
\end{compactenum}
\end{corollary}

\begin{definition}[Idempotence] \label{def:am_idempotent}
    We call an action model $(\calU, \alpha)$ \emph{idempotent} if{f} 
   pointed updates  $\bigl(\calM \odot \calU,\,\, (w,\alpha)\bigr)$ and  $\Bigl((\calM \odot \calU) \odot\calU,\,\, \bigr((w,\alpha),\alpha\bigr)\Bigr)$ are defined and isomorphic for any pointed  Kripke model $(\calM, w)$. 
\end{definition}

\begin{theorem} \label{thm:synth_ac_idempotent}
    For any DBI normal formula $\varphi$ and any pointed Kripke model $(\calM, v)$, $\calM \odot \calU, (v, 0)$ is idempotent.
\end{theorem}
\begin{proof}
    The updates $\calM \odot \calU, (v, 0)$ and $(\calM \odot \calU) \odot \calU, ((v, 0), 0)$ are always defined by Theorem \ref{thm:update_synth_succ}.

    Since $\calU_\varphi$ has an out-tree structure, we conduct the proof by induction over this structure beginning from the root $0$.
    \\
    \textbf{Induction hypothesis:} For every state $((w, \alpha), \beta)$ in $(\calM \odot \calU) \odot \calU$, $\beta = \alpha$.
    \\
    \textbf{Base case:} Point $(v, 0)$ is the only state that can be combined with $0$, since by Definition \ref{def:synth_am1} $0$ has no incoming edges in $\calU_\varphi$.
    \\
    \textbf{Induction step:} Suppose the IH holds for state $((w, \alpha), \alpha)$, but not for its $i$-accessible state $((x, \beta), \gamma)$, meaning $\gamma \ne \beta$.
    Note that this $i$-accessible state must exist, since  $(\calM \odot \calU) \odot \calU, ((v, 0), 0)$ could not be smaller than $\calM \odot \calU, (v, 0)$, as all preconditions in $\calU_\varphi$ are propositional by Definition \ref{def:synth_am1}.
    Since $((w, \alpha), \alpha) {R_i^\calU}^\calU ((x, \beta), \gamma)$ it must be the case that also $(w, \alpha) R_i^\calU (x, \beta)$ by Definition \ref{def:pointed_update}.
    Since $\varphi$ is DBI normal, every state in $\calU_\varphi$ has no more than one $j$-accessible state, hence $\beta$ is the only $i$-accessible state from~$\alpha$.
    Therefore if $((w, \alpha), \alpha) {R_i^\calU}^\calU ((x, \beta), \gamma)$ it must be that $\gamma = \beta$.
\end{proof}

\section{The Role of Privatization}
\label{sec:privatization}
The minimality of the update synthesis method we have described relies on a particular structure of the action model used for the update.
We call this structure \emph{privatized}.
In this section, we give an explicit definition and show why it serves the purpose of preserving beliefs whenever possible.
We begin with the notion of modal syntactic tree, which represents the nesting of modalities within a formula as a tree and provides the basis of the formal definition of privatization.
Note that the structure of the action model from Def.~\ref{def:synth_am1} (almost) exactly follows this tree structure already, so there was no need to introduce it explicitly.

\begin{definition}[Modal syntactic tree]
\label{def:Synthtree}
The \emph{modal syntactic tree} $\mathcal{T}_{\varphi}$ of a goal formula~$\varphi$ is an out-tree with a single unlabeled root, while all non-root nodes are labeled with modal operators.
\looseness=-1
\end{definition}

Some examples of modal syntactic trees are provided in Fig.~\ref{fig:mod_synt_trees}:
In  $\mathcal{T}_{B_i \xi}$, the root has one child-leaf labeled $B_i$.
Both  $\mathcal{T}_{B_i \varphi}$ and $\mathcal{T}_{B_i (\xi \wedge \varphi)}$ are obtained by labeling the root of $\mathcal{T}_{\varphi}$ with $B_i$ and making it the only child of the new root.
Finally, $\mathcal{T}_{\varphi \wedge \psi}$ is obtained by taking the disjoint union of  $\mathcal{T}_{\varphi}$  and~$\mathcal{T}_{\psi}$ and identifying their roots. 

\begin{figure}[t]
    \centering
    \resizebox{12cm}{!}{%

\begin{subfigure}{0.4\textwidth}
\begin{tikzpicture}
    \begin{scope}[xshift=-6cm,>=stealth, yshift=3cm]
    \node[] (0) {};
    \node[circle,draw, fill=black] (1) [above = 2.3cm  of 0] {};
    \node [left=.5cm of 1] {\textbf{\LARGE{$\mathcal{T}_{B_i \xi}$}}};
    \node[] (2) [below = 1.6cm  of 0] {};
    \node[circle, thin, draw, minimum size=8mm,inner sep=2pt] (2) [below = 1cm  of 1] {$B_i$};
    
            \path[->, -{Stealth[scale width=1.5]}]
        (1) edge node {} (2)
          ;
    \end{scope}
\end{tikzpicture}
\label{1}
\end{subfigure}

\hfill

\begin{subfigure}{0.7\textwidth}
\begin{tikzpicture}
    \begin{scope}[xshift=-6cm,>=stealth, yshift=3cm]
    \node[circle,draw, fill=black] (0) {};
    \node [left=.5cm of 0] {\textbf{\LARGE{$\mathcal{T}_{B_i \varphi} = \mathcal{T}_{B_i (\xi \wedge \varphi)}$}}};
    \node[circle, thin, draw, minimum size=8mm,inner sep=2pt] (2) [below = 1cm  of 0] {$B_i$};
    
            \path[->, -{Stealth[scale width=1.5]}]
        (0) edge node {} (2)
          ;
    \fill [fill=white!80!black]
    (0, -3.5) node {$\mathcal{T}_{ \varphi}$}
     (0,-2.05) node              {}
  -- (-1,-4.5) node             {}
  -- (1,-4.5) node              {}
  ;
  
    \end{scope}
\end{tikzpicture}
\label{5}
\end{subfigure}

\hfill

\begin{subfigure}{0.3\textwidth}
\begin{tikzpicture}
    \begin{scope}[xshift=-6cm,>=stealth, yshift=3cm]
    \node[circle,draw, fill=black] (0) {};
    \node [left=.5cm of 0] {\textbf{\LARGE{$ \mathcal{T}_{\varphi \wedge \psi}$}}};
    \node[] (3) [below = 3.7cm  of 0] {};

    \fill [fill=white!80!black]
    (-0.6, -1.75) node {$\mathcal{T}_{\varphi}$}
     (-0.1,-0.2) node              {}
  -- (-1.5,-2.65) node[behind path] {}
  -- (-0.1,-2.65) node              {}
  ;
  
  \fill [fill=white!80!black]
    (0.6, -1.75) node {$\mathcal{T}_{\psi}$}
     (0.1,-0.2) node              {}
  -- (1.5,-2.65) node[behind path] {}
  -- (0.1,-2.65) node              {}
  ;
    \end{scope}
\end{tikzpicture}
\label{5}
\end{subfigure}
    }
    \caption{Modal syntactic trees\looseness=-1}
    \label{fig:mod_synt_trees}
\end{figure}
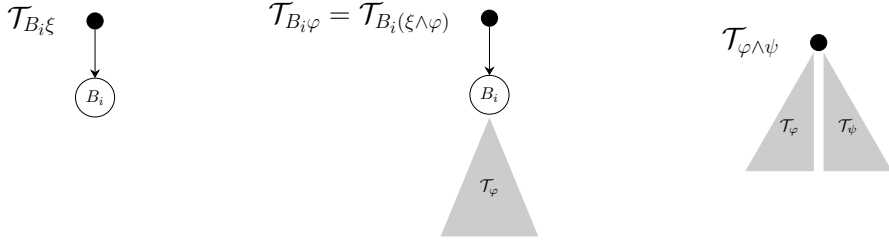

\begin{definition}[Modal-syntactic-tree root paths] \label{def:rootp}
\sloppy
    For formula $\varphi$ we define $\rootp(\varphi)$ as the set of paths $\big((root, \alpha_1, i_1), \allowbreak \ldots, (\alpha_{l-1}, \alpha_l, i_l)\big)$   of length $l\geq0$ in $\mathcal{T}_{\varphi}$  starting from the root, where $i_k$ is the label of $\alpha_k$ for $k=1,\dots,l$.
\end{definition}

\begin{definition}[Kripke-frame root walks] \label{def:rootw_nsr}
    For a pointed Kripke frame $(\langle S, R \rangle, w)$ we similarly define the set of root walks $\rootw(\langle S, R \rangle, w)$ as the set of all walks starting from the root (point) $w$.
    We further define the restriction $\rootwnsr(\langle S, R \rangle, w)$ as the subset of 
    $\rootw(\langle S, R \rangle, w)$, where we exclude walks that contain at least two successive edges for the same agent.
\end{definition}

\begin{definition}[Agent sequence]
\label{def:agseq}
	For a walk $\sigma = \allowbreak\bigl((\alpha_1, \alpha_2, i_1), \allowbreak \ldots, (\alpha_l, \alpha_{l+1}, i_l)\bigr)$ from Defs.~\ref{def:rootp}~or~\ref{def:rootw_nsr} we define
    $\agseq(\mathit{\sigma}) \ce (i_1,  \dots, i_l)$.
    In particular, $\agseq(\varepsilon) \ce \varepsilon$.
    We  extend this definition to sets, where for a set of walks $\Sigma$, $\agseq(\Sigma)$ is the set of corresponding agent sequences.
    \looseness=-1

    The set of all agent sequences of length $l$ we denote by $\agents^l$.
    The set of all agent sequences of length $l$ without any successively repeating agents we denote by $\agents^l_{\nsr}$.
\end{definition}

We now define agent-clusters and agent-accessible parts of the model.
\begin{definition}[Clusters and reachable states] \label{def:cluster}
For a Kripke frame $\calF = \langle W, R \rangle$, world $w\in W$, and sequence $(i_1, \ldots,  i_l)\in\agents^l$ of agents,
	we introduce
	the \emph{cluster of path-accessible worlds}
	\begin{equation}  \label{eq:agent_cluster}
			C_{\calF,w}^{i_1, i_2, \ldots, i_l} \ce 
			\bigl\{u \in W \mid (\exists u_2, 
			 \ldots, u_l \in W)\ wR_{i_1}u_2R_{i_2}   
			 \ldots u_l R_{i_l} u \bigr\}.
	\end{equation}
	In particular, for the empty sequence, $C_{\calF,w}^{\varepsilon}= \{w\}$.
\end{definition}

Finally, we have the necessary tools to formally describe the notion of privatization.
\begin{definition}[Privatized] \label{def:privatized}
    We call a pointed Kripke frame $(\calF, w) = (\langle S, R \rangle, w)$ \emph{privatized} w.r.t. formula $\varphi$ if{f} for any root path $\sigma \in \rootp(\varphi)$
    \begin{equation} \label{eq:ep_privat}
        C^{\agseq(\sigma)}_{\calF, w} \ne \varnothing \text{ and } \left(\forall agSeq \in \bigcup\limits_{l=0}^{\infty} \agents^l_{\nsr} \setminus \{\agseq(\sigma)\}\right)\ \left(C^{\agseq(\sigma)}_{\calF, w} \cap C^{agSeq}_{\calF, w}\right) = \varnothing.
    \end{equation}
\end{definition}

\begin{definition}[Weakly Privatized] \label{def:weakly_privatized}
    Similarly we call a pointed Kripke frame \mbox{$(\calF, w) = (\langle S, R \rangle, w)$} \emph{weakly privatized} w.r.t. formula $\varphi$ if{f} for any root path \mbox{$\sigma \in \rootp(\varphi)$}
    \begin{equation} \label{eq:ep_privat}
        \left(\forall agSeq \in \bigcup\limits_{l=0}^{\infty} \agents^l_{\nsr} \setminus \{\agseq(\sigma)\}\right)\ \left(C^{\agseq(\sigma)}_{\calF, w} \cap C^{agSeq}_{\calF, w}\right) = \varnothing.
    \end{equation}
\end{definition}

\begin{definition}[Walk accessibility] \label{def:wacc}
    For a pointed Kripke frame $(\calF, w) = (\langle W, R \rangle, w)$ we define \emph{walk accessibility operator}
    \begin{equation}
        \wacc(u, \calF, w) \ce \left\{\sigma \in \rootwnsr(\calF, w) \mid \pi_2 \pi_{|\sigma|}\sigma = u\right\}
    \end{equation}
    to be the set of agent-alternating walks starting from $w$ and ending in $u$.
\end{definition}

Walk accessibility yields an alternative equivalent definition of privatized frames:
\begin{corollary}[Privatized alternative] \label{cor:privat_via_wacc}
    A pointed Kripke frame $(\calF, w)$ is privatized w.r.t.~$\varphi$ if{f} for any root path $\sigma \in \rootp(\varphi)$
    \begin{equation} \label{eq:privat_via_wacc}
        C_{\calF, w}^{\agseq(\sigma)} \ne \varnothing \text{ and } \left(\forall s \in C_{\calF, w}^{\agseq(\sigma)}\right)\ \left|\agseq\bigl(\wacc(s, \calF, w)\bigr)\right| = 1.
    \end{equation}
\end{corollary}
\begin{proof}
    Follows from Defs. \ref{def:agseq}, \ref{def:privatized} and \ref{def:wacc}.
\end{proof}

\begin{corollary}[Weakly privatized alternative] \label{cor:weakly_privat_via_wacc}
    A pointed Kripke frame $(\calF, w)$ is privatized w.r.t.~$\varphi$ if{f} for any root path $\sigma \in \rootp(\varphi)$
    \begin{equation} \label{eq:privat_via_wacc}
        \left(\forall s \in C_{\calF, w}^{\agseq(\sigma)}\right)\ \left|\agseq\bigl(\wacc(s, \calF, w)\bigr)\right| = 1.
    \end{equation}
\end{corollary}
\begin{proof}
    Follows from Defs. \ref{def:agseq}, \ref{def:weakly_privatized} and \ref{def:wacc}.
\end{proof}

\begin{theorem} \label{thm:synth_ac_privatized}
    For DBI formula $\varphi$, the pointed action model $(\calU_\varphi, 0)$ is privatized w.r.t.~$\varphi$.
\end{theorem}
\begin{proof}
    by induction on the recursive construction of $\calU_\varphi$ as by Definition \ref{def:synth_am1}.
    \\
    \textbf{Base case:} For formula $B_i \xi$, where $\xi$ is propositional,
    since \mbox{$\agseq(\rootp(B_i \xi)) = \{\varepsilon, (i)\}$}, we only need to check all states walk-accessible via these two sequences. 
    Since in $\calU_{B_i \xi}$ the state $m \notin \{0, -1\}$ is the only state accessible via an $i$-edge from the root $0$, as the $0$ itself has no incoming arrows, $\agseq(\wacc(m, \calU_{B_i \xi}, 0)) = \{i\}$.
    Furthermore the root $0$ is only walk-accessible via the empty sequence $\varepsilon$, therefore $\agseq(\wacc(0, \calU_{B_i \xi}, 0)) = \{\varepsilon\}$.
    Lastly as the sink~$-1$ is not walk-accessible via any of the two sequences $\varepsilon$ and $(i)$ in $\agseq(\rootp(B_i \xi))$ we conclude that $(\calU_{B_i \xi}, 0)$ is indeed privatized w.r.t. $B_i \xi$.
    \\
    \textbf{Induction step:}
    If $\varphi = B_i \psi$ for DBI normal formula $\psi$, then $(\calU_\psi, 0)$ is privatized w.r.t. $\psi$.
\sloppy
    $\rootp(\varphi)$ extends the paths in $\rootp(\psi)$ right after the root with an edge to $m$ as so: 
    \begin{multline*}
    \rootp(\varphi) =  \{\varepsilon\}\cup\{((root, m, i), (m, \alpha_1, j_1)... (\alpha_{l-1}, \alpha_l, j_l)) \mid \\ ((root, \alpha_1, j_1), ... (\alpha_{l-1}, \alpha_l, j_l)) \in \rootp(\psi) \text{ and } l \ge 0\} .
    \end{multline*}
    Since the set of walks in $\rootwnsr(\calU_\psi, 0)$ is extended in exactly the same way, $(\calU_\varphi, 0)$~must also be privatized w.r.t. $\varphi$.

    If $\varphi = B_i(\xi \wedge \psi)$ for propositional formula $\xi$ and DBI normal formula $\psi$, the reasoning is analogous to the previous case, as privatization is a frame property and thus independent of the additional precondition added.

    If $\varphi = \psi \wedge \theta$, for DBI normal formulas $\psi$ and $\theta$, then 
    \[
    \rootp(\varphi) = \rootp(\psi) \cup \rootp(\theta).
    \]
    Again since by Definition \ref{def:synth_am1} $\rootw(\calU_\varphi) = \rootw(\calU_\psi) \cup \rootw(\calU_\theta)$ is extended the same way, $(\calU_\varphi, 0)$ is again privatized w.r.t. $\varphi$ and we are done.
    Note that in $\calU_\varphi$ not only the root $0$, but also the sink nodes $-1$ of $\calU_\psi$ and $\calU_\theta$ are unified.
\end{proof}

We will sometimes lightly abuse the notation by applying properties defined on Kripke frames to Kripke models $\langle S, R, V\rangle$ and action models  $\langle E, Q, \pre\rangle$. In all such cases, we mean the property to be applied to frames $\langle S, R\rangle$ and  $\langle E, Q\rangle$ respectively.

\begin{lemma} \label{lem:rootp_up_equiv_act}
    Let the pointed update $\bigl(\calM\odot\calU,(w,\alpha)\bigr)$ of a pointed Kripke model $(\calM,w)$ with a pointed action model $(\calU,\alpha)$ be defined. For any $agSeq \in \bigcup\limits_{l=0}^{\infty} \agents^l$,
    \begin{equation}
        (x, \beta) \in C^{agSeq}_{\calM \odot \calU, (w, \alpha)} \qquad \Longrightarrow \qquad \beta \in C^{agSeq}_{U, \alpha} \text{ and } x \in C^{agSeq}_{\calM, w}.
    \end{equation}
\end{lemma}
\begin{proof} Let  $\calM = \langle S, R, V \rangle$, $\calU = \langle E, Q, pre \rangle$, and $\calM\odot\calU = \langle S', R', V'\rangle$. We use induction on $l = |agSeq|$. \\
    \textbf{Base case $l = 0$:}
    Since $\agents^0 = \{\epsilon\}$ the statement follows trivially from Defs.~\ref{def:cluster}~and~\ref{def:pointed_update}.
\\
    \textbf{Induction step:}
    Consider any agent sequence $\overline{agSeq} = agSeq \circ i$ of length $l+1$ and state $(x, \beta) \in C^{\overline{agSeq}}_{\calM \odot \calU, (w, \alpha)}$.
    By Def.~\ref{def:cluster}, there must exist a state $(v, \gamma) \in C^{agSeq}_{\calM \odot \calU, (w, \alpha)}$ such that $(v, \gamma) R'_i (x, \beta)$, in particular, $v R_i x$ and $\gamma Q_i \beta$ by Def.~\ref{def:pointed_update}. By IH for $agSeq$ of length~$l$, we have $v \in C^{agSeq}_{\calM, w}$ and  $\gamma \in C^{agSeq}_{\calU, \alpha}$. Thus, by Def.~\ref{def:cluster} both $x \in C^{\overline{agSeq}}_{\calM, w}$ and $\beta \in C^{\overline{agSeq}}_{\calU,\alpha}$.
\end{proof}

\begin{theorem} \label{thm:update_with_privatized_ac_leads_to_weakly_privatized_em}
    For pointed action model $(\calU, \alpha)$ privatized w.r.t. DBI formula $\varphi$ and pointed epistemic Kripke model $(\calM, w)$ if their pointed update exists, then $(\calM \odot \calU, (w, \alpha))$ is weakly privatized w.r.t. $\varphi$.
\end{theorem}
\begin{proof}
    Suppose by contradiction that the pointed action model $(\calU, \alpha)$ is privatized w.r.t. $\varphi$, but for some pointed Kripke model $(\calM, v)$, the update $(\calM \odot \calU, (w, \alpha))$ exists, however is not weakly privatized w.r.t. $\varphi$.
    This means that there exists a root path $\sigma \in \rootp(\varphi)$,
    agent sequence $agseq \in \bigcup\limits^{\infty}_{l=0} \agents^l_{nsr} \setminus \agseq(\sigma)$ and
    state \mbox{$(x, \beta) \in C^{agseq}_{\calM \odot \calU, (w, \alpha)}$ and $(x, \beta) \in C^{\agseq(\sigma)}_{\calM \odot \calU, (w, \alpha)}$}.
    By Lemma \ref{lem:rootp_up_equiv_act} this implies that \mbox{$\beta \in C^{\agseq(\sigma)}_{\calU, \alpha}$} and $\beta \in C^{agSeq}_{\calU, \alpha}$.
    However this contradicts the assumption that $(\calU, \alpha)$~is privatizing w.r.t. $\varphi$, as $C^{\agseq(\sigma)}_{\calU, \alpha} \cap C^{agSeq}_{\calU, \alpha} = \varnothing$ 
\end{proof}

\begin{theorem} \label{thm:synth_ac_privatizing}
    For DBI formula $\varphi$ and any pointed Kripke model $(\calM, v)$, the pointed update between $(\calU_\varphi, 0)$ and $(\calM, v)$ is defined and $(\calM \odot \calU, (w, \alpha))$ is weakly privatized.
\end{theorem}
\begin{proof}
    Follows immediately from Theorem \ref{thm:update_with_privatized_ac_leads_to_weakly_privatized_em}.
\end{proof}
\section{Conclusions}
\label{subsec:concl}

We presented a way to synthesize a pointed action model for a limited range of goal formulas that:
(i) makes the goal formula true in the resulting model;
(ii) privatizes any Kripke model to which it is applied to, thus breaking the common knowledge of the model assumption and simulating totally private communication;
(iii) preserves consistency whenever possible, marking a difference from, e.g.,~\cite{AUL-synth}; and
(iv) has minimal side effects, in the sense that it changes as few beliefs unspecified in the goal formula as possible.
In addition, the synthesized pointed action models are combined with pointed Kripke models via the new pointed update operation, which does not apply to the whole model globally, but rather it respects the structure of the privatization introduced in the synthesis.
Since the pointed update is a subset of the full product update operation, the proposed update mechanism is quite efficient: only the initial  update increases the size of the model at worst linearly in the size of the goal formula, while all subsequent updates via goal formulas with the same  belief structure (or substructure thereof) only decrease the model size. Compare this with the exponential blow up after iterated updates using the standard product updates.

\emph{Future work.} 
We aim at extending the update mechanism so as to allow a wider range of goal formulas, introducing negations of modal operators, which might introduce ignorance, and disjunctions of modal operators, which have multiple possible realizations, while preserving the properties of privatization, minimality and consistency.
Furthermore, we plan to introduce group updates such as public announcements, providing full granularity in the update design.

\emph{Acknowledgments.}  We are grateful to Hans van Ditmarsch, Stephan Felber, Kristina Fruzsa, Rojo Randrianomentsoa, Hugo Rinc\'on Galeana, and Ulrich Schmid for multiple illuminating and inspiring discussions.

\bibliographystyle{abbrvurl}
\bibliography{references1}

\end{document}